\documentclass[a4paper,10pt]{article}
\usepackage{booktabs} 
\usepackage[ruled,linesnumbered]{algorithm2e} 

\SetAlFnt{\small}
\SetAlCapFnt{\small}
\SetAlCapNameFnt{\small}
\SetAlCapHSkip{0pt}
\IncMargin{-\parindent}

\usepackage{natbib} 
\usepackage[margin=1in]{geometry}

\setcitestyle{authoryear}

\title{Fair Online Resource Allocation}

\usepackage{authblk}

\author[1]{Christopher En\thanks{Corresponding Author: ce2456@columbia.edu}}
\author[1]{Yuri Faenza\thanks{yf2414@columbia.edu}}
\author[2]{Andrea Lodi\thanks{al748@cornell.edu}}
\author[3]{Gonzalo Mu\~noz\thanks{gonzalo.m@uchile.cl}}
\affil[1]{Columbia University, IEOR Department}
\affil[2]{Cornell Tech}
\affil[3]{Universidad de Chile}

\usepackage[utf8]{inputenc}
\usepackage{amsthm}
\usepackage{amsmath}
\usepackage{amsfonts}
\usepackage{float}

\usepackage{thmtools, thm-restate}

\usepackage{graphicx}
\usepackage{subcaption}

\usepackage{xcolor}

\usepackage{hyperref}
\hypersetup{colorlinks=true,linkcolor=blue,citecolor=blue,urlcolor=blue}

\usepackage{soul}

\usepackage{color-edits}
\addauthor{ce}{blue}
\addauthor{yf}{purple}
\addauthor{al}{red}
\addauthor{gm}{green}

\usepackage[normalem]{ulem}

\usepackage{enumerate}

\newtheorem{myclaim}{Claim}
\theoremstyle{plain}
\newtheorem{theorem}{Theorem}

\newtheorem{corollary}[theorem]{Corollary}
\newtheorem{lemma}[theorem]{Lemma}

\theoremstyle{definition}
\newtheorem{definition}[theorem]{Definition}
\newtheorem{example}{Example}[section]

\newtheorem*{theorem*}{Theorem}

\theoremstyle{remark}

\renewcommand{\[}{\begin{equation}}
\renewcommand{\]}{\end{equation}}

\newcommand{\bea}{\begin{eqnarray}}
\newcommand{\eea}{\end{eqnarray}}

\newcommand\gfair{\mathop{\mbox{$\gamma$-$\mathit{FAIR}$}}}

\newcommand\unf{\mathop{\mbox{$\mathit{UNF}$}}}

\newcommand\flu{\mathop{\mbox{$\mathit{FLU}$}}}

\newtheoremstyle{boldnum}
  {\topsep}   
  {\topsep}   
  {\itshape}  
  {}          
  {\bfseries} 
  {.}         
  {.5em}      
  {#1 #3}     

\theoremstyle{boldnum}
\newtheorem*{manualthm}{Theorem}

\begin{document}

\maketitle

\begin{abstract}
We study the problem of fair online resource allocation, motivated by applications such as refugee resettlement and airline scheduling, where agents arrive sequentially and must be assigned to facilities with limited capacities. We introduce a model that maximizes the overall welfare subject to resource constraints and a Lipschitz fairness requirement, which ensures that similar agents arriving in the same batch receive similar expected outcomes. We first analyze the offline problem, proving that the value of the optimal fair allocation is at least an $\Omega(1/\gamma)$ fraction of the optimal unfair allocation, where $\gamma$ is the fairness coefficient, thereby bounding the price of fairness. For the online setting, we propose an algorithm based on dual mirror descent that enforces fairness constraints within batches while estimating optimal dual variables. We prove that this algorithm achieves sublinear regret relative to the optimal offline fluid benchmark. Finally, we validate our theoretical results using real-world data from the Refugee Economies Programme, demonstrating the algorithm's performance and examining the trade-offs between welfare maximization and fairness enforcement.
\end{abstract}

\section{Introduction}\label{sec:introduction}

The problem of allocating limited resources to agents arriving sequentially is a fundamental challenge in operations research and economics, appearing in domains ranging from online advertising and airline scheduling to public sector applications like housing allocation. In many of these settings, the primary objective is to maximize a global welfare metric—such as revenue, efficiency, or social utility—subject to capacity constraints. For instance, in online advertising auctions, algorithms aim to allocate ad slots to advertisers to maximize total revenue while respecting budget constraints (e.g., the AdWords problem, see~\cite{mehta2010online}). Similarly, in airline revenue management, the goal is often to price and allocate seats to maximize yield given fixed capacity, see, e.g.,~\cite{talluri2006theory}.

In these scenarios, standard algorithms—such as those based on online primal-dual formulations or dual mirror descent—focus purely on the aggregate value, effectively treating agents (advertisers or passengers) as means to an optimal end. However, in many high-stakes applications, such as those involving human subjects, maximizing utilitarian welfare alone is often insufficient. This tension has sparked significant research into the ``price of fairness'': quantifying the loss in system efficiency required to achieve equitable outcomes. This trade-off is not limited to human-centered domains; it appears in computer networking (e.g., fair queuing to prevent bandwidth starvation) \citep{demers1989analysis, kleinberg1999fairness}, machine learning (ensuring classifiers do not discriminate against protected groups) \citep{dwork2012fairness, hardt2016equality}, and organ exchange programs, where allocating kidneys purely to maximize life-years might systematically disadvantage hard-to-match patients \citep{dickerson2014price}. In these settings, fairness is not just a constraint but a fundamental requirement of system stability and ethics.

 An example motivating our work is refugee resettlement, where families arrive sequentially and must be assigned to host communities with limited capacity. Existing approaches, including bespoke refugee resettlement algorithms such as those proposed by \citet{bansak2018improving} and \citet{ahani2021placement}, explicitly aim to maximize global expected employment. While well-intentioned, unconstrained optimization in this context can lead to severe inequities. Consider a scenario where two refugees with nearly identical profiles arrive. A profit-maximizing algorithm, prioritizing agents who contribute most to the objective, might assign the first refugee to a highly desirable urban center where they have a high chance of employment. However, to conserve capacity for a hypothetical ``perfect match'' arriving later, the algorithm might relegate the second, nearly identical peer to a remote area with scarce opportunities. Despite their similarities, their outcomes diverge drastically based on minor profile differences or mere arrival timing. Such disparities violate the principle of individual fairness,  which dictates that similar individuals should be treated similarly. See Example~\ref{ex:fairness-matters} for details.  Interestingly, in Example~\ref{ex:fairness-matters} there exists a solution that is at the same time ex-ante fair, and whose expected objective function value is close to that of the optimal solution that disregards fairness concerns.   This is the solution which assigns to the two agents similar probability vectors of being matched to each of the two locations, and then samples an assignment from these probability distributions. See again Example~\ref{ex:fairness-matters} for details.  
 Drawing inspiration from such an example, this work seeks to answer the following research questions: 
 
 \begin{quote}\emph{Is the ``price of fairness'' always bounded\footnote{In this paper, we focus on ex-ante fairness. As is common for models with indivisible goods, ex-post fairness can only be guaranteed in the model we consider if we accept to drastically reduce the objective function value.} for online resource allocation problems? If so, can we efficiently find a fair solution that achieves high profit in expectation? }
\end{quote}

In this paper, we give positive answers to both questions above and more generally address the tension between welfare maximization and fairness in an online resource allocation model. We consider a setting  with groups of agents arriving simultaneously as batches at discrete time steps (e.g., a weekly intake of asylum seekers), rather than one by one. As a batch arrives, we want to find a lottery for assigning its agents to a set of facilities, so that the resources consumed by agents when the lottery is realized do not exceed the resource budget. To maximize profit, we select the lotteries so as to optimize a linear objective function defined over pairs of agents and facilities. To guarantee fairness, we assume to be also given a distance function between agents, and constrain the selected lotteries to satisfy a Lipschitz  constraint between any pair of agents in the same batch, ensuring that any difference in expected outcomes between two agents is bounded proportionally to how different the agents are. If two agents are indistinguishable, they must receive statistically identical assignments. See Section~\ref{sec:preliminaries} for a formal description of the model. Compared to offline fairness constraints, enforcing this condition in an online setting is highly non-trivial. The decision-maker must satisfy these local fairness constraints within each batch to avoid discrimination, while simultaneously managing global resource consumption over the entire horizon to avoid running out of capacity. A key challenge we address is determining how detrimental this fairness requirement is to the global objective function.


\subsection*{Main Contributions}
Our work provides a rigorous theoretical framework as well as a sound empirical analysis for fair online allocation. Our contributions are three-fold:

\noindent \paragraph{1. Bounding the Price of Fairness (Offline Analysis).} We first analyze the offline version of the problem, where all arrivals are known in advance. We prove that the optimal fair allocation is not arbitrarily worse than a fluid benchmark for the unfair allocation. In the following theorem, $\gamma$ is a parameter determining the fairness requirements of the model. The larger the $\gamma$, the more stringent the fairness constraints imposed on the lotteries.

\begin{theorem*}[Informal version of Theorem~\ref{thm:fairness}]
    For any constant $\gamma\ge 1/2$, the value of the offline $\gamma$-fair resource allocation problem is at least an $\Omega(1/\gamma)$ fraction of the value of the offline fluid unfair allocation problem.
\end{theorem*}

Theorem~\ref{thm:fairness} provides a theoretical guarantee on the ``price of fairness,'' demonstrating that equitable outcomes can be achieved without a catastrophic loss in total welfare.

\noindent \paragraph{2. Online Algorithm with Sublinear Regret.} We propose a new online algorithm based on dual mirror descent. Our approach adapts the standard dual descent framework by enforcing Lipschitz fairness constraints within the primal allocation step for each batch. We prove that this algorithm effectively learns the optimal dual variables (shadow prices) over time, achieving sublinear regret.
\begin{theorem*}[Informal version of Theorem~\ref{thm:regret}]
    There exists an online algorithm that achieves $O(\sqrt{T})$ expected regret for the fair online resource allocation problem for constant batch size, where $T$ is the time horizon.
\end{theorem*}

To our knowledge, this is one of the first results establishing sublinear regret for online allocation under Lipschitz fairness constraints.

\paragraph{3. Empirical Validation on Refugee Data.} We validate our theoretical findings using real-world data from the Refugee Economies Programme \citep{betts2024economic}, comprising household survey data from Kenya, Uganda, and Ethiopia. Our experiments demonstrate that the proposed algorithm performs efficiently in practice, achieving welfare close to the offline optimum while strictly satisfying fairness requirements. In contrast, we observe empirically on this dataset that, in the allocation produced by the algorithm by~\citet{balseiro2020dual}, a significant proportion of agents are treated unfairly (according to our criterion), thus justifying the rationale for imposing the fairness constraint explicitly. We also provide a detailed analysis of the trade-offs, showing that in many realistic regimes, significant fairness gains can be realized with only a marginal reduction in total employment outcomes.

\subsection*{Technical Overview}

Our theoretical analysis addresses two primary challenges: quantifying the efficiency loss inherent to fairness (bounding the ``Price of Fairness'') in the offline setting, and designing an algorithm that learns to allocate optimally under these strict constraints in the online setting (minimizing the ``Price of Learning''). 

\smallskip

\noindent\textbf{Bounding the Price of Fairness}. The first major challenge is to determine how much global welfare must be sacrificed to satisfy the Lipschitz fairness constraint. The optimal fair solution is the solution to a large linear program with numerous fairness constraints coupling the agents, making direct analysis of the feasible region difficult. Instead of analyzing the fair solution space directly, our strategy is constructive. We derive a lower bound on the fair optimal value by constructing a specific allocation that is both provably fair and a constant approximation of the fluid unfair solution. We introduce a ``Fair Water-Filling'' algorithm to generate this candidate solution. In this algorithm, we initially have agents fill ``buckets'' of probability mass in accordance with the fluid unfair solution. As fairness constraints become tight, the algorithm dynamically adjusts the buckets to maintain a baseline of fairness while simultaneously maintaining an approximation of the fluid unfair solution and ensuring resource constraints remain satisfied. We prove that this process terminates with a solution that is fair and achieves a constant fraction of the unfair welfare, establishing that the Price of Fairness is bounded (Theorem~\ref{thm:fairness}).

\smallskip

\noindent\textbf{Minimizing the Price of Learning}. The second challenge is the online setting, where we must make irrevocable allocations batch-by-batch without knowledge of future arrivals. The core difficulty here is satisfying the local fairness constraints (fairness within the current batch) while optimizing the global resource constraints (capacity over the entire horizon). We adapt a framework based on the dual online mirror descent algorithm of \citet{balseiro2020dual}. The standard approach in online allocation is to learn optimal dual variables (shadow prices) for the resources. However, a standard greedy assignment based on these prices would violate the fairness constraints. Instead of the myopic greedy approach, we solve a local fair optimization problem for each batch. Thus, the algorithm balances welfare against the estimated shadow cost of resources, showing Theorem~\ref{thm:regret}.

\subsection*{Related Work}

Online resource allocation has been one of the most fundamental and well-studied problems in operations research, computer science, and economics. The vast majority of literature in the field has been purely utilitarian, focusing on maximizing some central or overall objective. \citet{devanur2009adwords} studied the AdWords problem, where the value derived from an allocation is proportional to the resources allocated, presenting an algorithm which achieved $O(T^{2/3})$ regret. \citet{agrawal2014dynamic} present a more complex algorithm that involves solving successive linear programs, achieving $O(T^{1/2})$ regret. \citet{agrawal2014fast} also provide an algorithm for a more general class of online convex programming problems, achieving $O(T^{1/2})$ regret, though they allow some constraint violations. \citet{guo2022online} and \citet{castiglioni2022unifying} have similar results with similar constraint tradeoffs. \citet{balseiro2020dual} adapt the online mirror descent algorithm from optimization literature to learn optimal dual ``shadow prices,'' achieving $O(T^{1/2})$ regret in more general online resource allocation settings. Our algorithm is based on their techniques, adapted to incorporate fairness constraints. A common theme throughout these works is a focus purely on utility maximization, without consideration for fairness over outcomes.

Fairness in algorithmic decision-making can generally be classified into two types: group fairness and individual fairness. Group fairness generally aims to ensure that various protected groups are treated equally, using metrics such as demographic parity or equality of opportunity. \citet{manshadi2021fair, donahue2020fairness, freund2023group, balseiro2021regularized} explore various revenue management problems under a range of group fairness constraints. However, group fairness faces some ethical limitations, such as ``fairness gerrymandering,'' where a subgroup within a protected class is treated poorly to balance aggregate statistics. In general, group fairness requirements may fail to treat individuals according to their specific merits and needs.

Individual fairness has also been studied extensively. \cite{esmaeili2023rawlsian} provide an online bipartite matching algorithm that simultaneously optimizes welfare and Rawlsian (max-min) individual fairness. \cite{sinclair2022sequential} study an online resource allocation problem where agent utility is proportional to resources consumed. They define fairness as a version of envy-freeness, and design an algorithm that achieves the optimal envy-efficiency tradeoff. However, they assume full knowledge of the arrival distribution. \citet{dwork2012fairness} introduce the concept of individual Lipschitz fairness in the context of machine learning classification, with the core axiom that ``similar individuals should be treated similarly.'' \citet{gupta2021individual} apply individual Lipschitz fairness to an online stochastic contextual bandit problem, requiring fairness over the ex-post decisions. In many resource allocation settings with indivisible goods, though, ex-post individual fairness is unreasonable or impossible to guarantee. Our work builds on the foundation set by \citet{dwork2012fairness}, adapting ex-ante individual Lipschitz fairness to the online resource allocation setting.

\subsection*{Organization}
The remainder of this paper is organized as follows. Section~\ref{sec:preliminaries} formally defines the model, the Lipschitz fairness constraint, and the fluid benchmarks. It also presents some examples to discuss the importance of explicitly imposing fairness constraint, as well as of some of our assumptions. Section~\ref{sec:offline} presents the offline analysis and the bound on the price of fairness, leading to a proof of Theorem~\ref{thm:fairness}. Section~\ref{sec:algorithm} introduces our fair online algorithm based on dual mirror descent, followed by the regret analysis in Section~\ref{sec:online}. These results prove Theorem~\ref{thm:regret}. Finally, Section~\ref{sec:experiments} presents experimental results using the refugee resettlement dataset.

\section{Preliminaries}\label{sec:preliminaries}

In this section, we present the model for fair online resource allocation. Our goal is to allocate facilities to incoming agents over a finite time horizon to maximize the overall weight of the matching, subject to resource constraints and a Lipschitz fairness requirement.

\paragraph{Input.} 
The primary components of our model are formalized as follows:
\begin{itemize}
    \item \textbf{Facilities and Resources}: We are given a finite set $V$ of distinct facilities and a set of $N$ different resources. 
    
    \item \textbf{Time and Arrivals}: Agents arrive over a discrete time horizon of $T$ periods, indexed by $t\in [T]$. In each period $t$, a batch of $S$ agents arrives, indexed by $s\in [S]$. Thus, each agent is uniquely identified by the pair of indices $t,s$. These agents are independently and identically distributed (i.i.d.) from an unknown distribution $F$ from a known family $\mathcal F$ over a known finite set of distinct agent types, denoted by $U$. We denote the type of the $s$-th agent that arrives in period $t$ by $u_{ts}\in U$.
    
    \item \textbf{Resource Capacity}: The total available capacity for all resources over the entire horizon is given by $B_0\in \mathbb N^{N}$. We also define $\rho := B_0/(TS)$. Note that resources are global and not tied to specific facilities. Moreover, each facility has the same initial amount of each resource, though it can be consumed at different rates, depending on the facility, the resource, and the agent -- see below.
    
    \item \textbf{Agents}: Each agent in each batch with type $u$ is defined by three characteristics $(w_u, b_u, e_u)$:
    \begin{enumerate}
        \item A \textbf{weight vector} $w_u \in [0, 1]^V$, where $(w_u)_v$ is the value obtained if agent $u$ is allocated to facility $v$. We assume $w_u > 0$ for all $u \in U$ (componentwise).
        \item A \textbf{resource consumption matrix} $b_u \in \mathbb{N}^{V\times N}$, representing the quantity of each of the $N$ resources consumed by an agent of type $u$ when allocated each facility in $V$.
        \item A unit \textbf{type vector} $e_u\in \{0,1\}^U$ indicating the type of the agent.
    \end{enumerate}
    The type of the $s$-th agent that arrives in period $t$ is given by $u_{ts}$. For notational convenience, we denote $(w_{ts}, b_{ts}, e_{ts}):= (w_{u_{ts}}, b_{u_{ts}}, e_{u_{ts}})$. Note that for algorithmic purposes, it is not necessary to define the types $U$ in advance.
  
    \item \textbf{Distance Function}: We are given a symmetric distance metric $d:\{0,1\}^U \times \{0,1\}^U \rightarrow \mathbb{R}_{\geq 0}$ measuring the dissimilarity between agent types. This metric is used to enforce fairness by ensuring that similar individuals are treated similarly. 
\end{itemize}

\paragraph{Allocations.} The decision at each time $t$ is how to allocate the arriving agents: each agent in the batch is assigned to at most one facility, represented by $\tilde x_{ts}\in \{0,1\}^V$. Decisions may also be randomized, where each agent is given a lottery $x_{ts}\in \{x\in \mathbb R^V_+\mid \sum_{v \in V} x_{v}\le 1\}$ over the facilities. Assignments are then realized as $\tilde x_{ts}\sim x_{ts}$ after allocation distributions are decided but before the next batch $t+1$ arrives. Define $\theta$ as the random variable determining the realization of $\tilde x_{ts}\sim x_{ts}$ for each agent in each batch.

\paragraph{Constraints.}
The allocation decisions must satisfy the following constraints:
\begin{enumerate}
    \item \textbf{Lipschitz Fairness Constraint}: For any two agents $s_1$ and $s_2$ arriving in the same batch $t$, the difference in their expected outcomes is bounded by their similarity:
    \begin{equation}\label{eq:lipsch-fair}
    d(u_{ts_1}, u_{ts_2}) \ge \gamma (\mathbb E_\theta[w_{ts_1}^\top \tilde x_{ts_1}] - \mathbb E_\theta[w_{ts_2}^\top \tilde x_{ts_2}]) = \gamma(w_{ts_1}^\top x_{ts_1} - w_{ts_2}^\top x_{ts_2}) \quad \forall t \in [T], \forall s_1, s_2 \in [S].
    \end{equation}
    Here, $\gamma > 0$ is a parameter controlling the strictness of the fairness requirement. An allocation satisfying this constraint, or an algorithm producing such allocations, is called \emph{$\gamma$-fair}. Note that this constraint corresponds to ex-ante fairness, instead of ex-post fairness; we discuss the rationale behind this choice below.
    
    \item \textbf{Resource Constraint}: The total consumption of resources across all agents and all time periods cannot exceed the available capacity:
    \begin{equation}\label{eq:resource-ex-post}
    \sum_{t=1}^{T} \sum_{s=1}^{S}  b_{ts}^\top \tilde x_{ts}  \le B_0.
    \end{equation}
    Note that even with randomized algorithms, the resource capacity constraint must be satisfied deterministically.

    \item \textbf{Allocation Feasibility}: Each agent can be assigned to at most one facility, namely
    \begin{align*}
        \sum_{v\in V} (\tilde x_{ts})_v &\le 1 \quad \forall t \in [T], s \in [S], \\
        \tilde x_{ts} &\ge 0 \quad \forall t \in [T], s \in [S].
    \end{align*}
\end{enumerate}

\paragraph{Reward of the online problem.}
Given an algorithm $A$ that allocates agents to facilities respecting the constraints above and an arrival distribution $F$, the reward of the algorithm over the distribution is given by
\[R(A\mid F) = \mathbb E_{F}\left[\sum_{t=1}^T\sum_{s=1}^S w_{ts}^\top \tilde x_{ts}\right],\]
where $\tilde x$ is computed by the algorithm.

\paragraph{Benchmarks.} The optimal fair reward, given the distribution $F$ of the arrivals $u_{ts}$ for each agent $s\in [S]$ in each batch $t\in [T]$, is bounded above by the following linear program:
\begin{equation*}
\gfair(\{u_{ts}\}_{t\in [T], s\in [S]}) =  \left[
\begin{aligned}
    \max_{x} \quad & \sum_{t=1}^T\sum_{s=1}^S w_{ts}^\top  x_{ts} \\
    \text{s.t.} \quad & \sum_{t=1}^T\sum_{s=1}^S b_{ts}^\top  x_{ts}\le B_0 \\
    & d(u_{ts_1}, u_{ts_2}) \ge \gamma (w_{ts_1}^\top  x_{ts_1} -  w_{ts_2}^\top  x_{ts_2}) && \forall t \in [T], \forall s_1, s_2 \in [S] \\
    & \sum_{v\in V} ( x_{ts})_v \le 1 && \forall t \in [T], s \in [S] \\
    &  x_{ts} \ge 0 && \forall t \in [T], s \in [S]
\end{aligned}
\right].
\end{equation*}
Note that the optimal solution to $\gfair$ may be fractional. We may interpret the optimal solution $x$ as a randomized allocation, which is required to satisfy the resource constraints in expectation only. The optimal offline unfair reward is bounded above by a similar linear program, without the Lipschitz fairness constraints:
\begin{equation*}
\unf(\{u_{ts}\}_{t\in [T], s\in [S]}) =  \left[
\begin{aligned}
    \max_{ x} \quad & \sum_{t=1}^T\sum_{s=1}^S w_{ts}^\top  x_{ts} \\
    \text{s.t.} \quad & \sum_{t=1}^T\sum_{s=1}^S b_{ts}^\top x_{ts}\le B_0 \\
    & \sum_{v\in V} ( x_{ts})_v \le 1 && \forall t \in [T], s \in [S] \\
    &  x_{ts} \ge 0 && \forall t \in [T], s \in [S]
\end{aligned}
\right].
\end{equation*}
We will refer to $\gfair$ and $\unf$ as the \emph{fluid relaxations}. Observe that for each realization of $F \in \mathcal{F}$ we have
\begin{equation}\label{eq:unflu}
\unf(\{u_{ts}\}_{t\in [T], s\in [S]})\geq \gfair(\{u_{ts}\}_{t\in [T], s\in [S]}).
\end{equation}

Similarly, the optimal \emph{ex-ante} reward for the fair problem is bounded by the expected optimal value of $\gfair$, over the arrival distribution $F$

\begin{equation}\label{eq:flu}
\flu(F) := \mathbb E_F\left[\gfair(\{u_{ts}\}_{t\in [T], s\in [S]})\right] \geq R(A\mid F).
\end{equation}

\paragraph{Objective function of the online problem.} The objective of the online problem is to minimize the regret of $A$ over $F$, relative to the fluid benchmark
\[Regret(A\mid F) = \flu(F) - R(A\mid F).\]
If $A$ is randomized, then the expected regret is given by
\[\mathbb E_\theta[Regret(A\mid F)] = \flu(F) - \mathbb E_\theta[R(A\mid F)].\]

\paragraph{Assumptions and Notation.}
We introduce the following notation and make the following assumptions on the model parameters and the distance function:
\begin{itemize}
    \item We denote $\overline w := \max_{u\in U}||w_u||_\infty$.
    \item We denote $\overline b := \max_{u\in U}||b_u||_\infty$.
    \item We assume there exists $\underline\rho, \overline\rho\in \mathbb R_{++}$ such that $\underline \rho \le \rho_{k}\le \overline \rho$ for all $k\in N$.
    \item \textbf{Distance Function Assumptions}:
    \begin{enumerate}
        \item[(0)] Every type has distance $0$ from itself, i.e., $d(u,u)=0$ for all $u \in U$. \item The triangle inequality holds: $d(u_1, u_2) + d(u_2, u_3)\ge d(u_1, u_3)$ for all $u_1, u_2, u_3\in U$.\label{d_assumption_t}
        \item The distance is lower bounded by weight vector difference: $d(u_1, u_2)\ge ||w_{u_1} - w_{u_2}||_\infty$ for all $u_1, u_2\in U$.\label{d_assumption_w}
        \item If resource consumption differs, distance is bounded away from zero: If $b_{u_1}\ne b_{u_2}$, then $d(u_1, u_2)\ge \underline d$, for some constant $\underline d>0$ independent of $u_1,u_2$. \label{d_assumption_b}
    \end{enumerate}
\end{itemize}

\paragraph{Discussion.}
\begin{itemize}
\item \textbf{Fairness Constraint:} Constraint~\eqref{eq:lipsch-fair} is an ex-ante condition on the probability $x$ that similar agents (as indicated by the distance function $d$) are assigned to facilities that guarantee a similar value (as indicated by the vector $w$). Once the assignment is realized, similar agents may obtain very different values. An ex-post fairness condition is however not reasonable to impose, as it is easy to see that in some cases it can only be guaranteed by leaving all agents unmatched.
\item \textbf{Resource Constraint:} Constraint~\eqref{eq:resource-ex-post} is an ex-post condition on the realized allocation $\tilde x$. Note that the benchmarks are only required to satisfy the corresponding constraint in expectation over $x$. 
\item \textbf{Assumptions:}
In general, the assumptions we make on the setting are standard and nonrestrictive (see, e.g., \citet{balseiro2020dual,bansak2025dynamic}). For the distance function, we connect to the input weight vector and resource consumption vectors to ensure that agents with similar characteristics must be given similar outcomes; in this way, we can prevent unfair discrimination based on outside factors. 

We also relate the distance function to the $\ell_\infty$-norm of the weight vector as a matter of sanity: alternatives such as the $\ell_1$- and $\ell_2$-norms encounter issues with dimensionality, where as the number of facilities increases, the distance between similar and dissimilar agents becomes less differentiated. Furthermore, the $\ell_1$- and $\ell_2$-norms are not robust to symmetries such as splitting a facility into two identical facilities with half capacity. For instance, consider a setting with a single facility, and two types of agents with weight vectors $w_{u_1} = [0.9]$ and $w_{u_2} = [0.5]$. The assumptions on the distance function require that $d(u_1, u_2)\ge ||w_{u_1} - w_{u_2}||_\infty = 0.4$. If we create another instance by splitting the facility into two identical facilities with half as many resources, the agents have weight vectors $w'_{u_1} = [0.9\ 0.9]$ and $w'_{u_2} = [0.5\ 0.5]$. The distance function constraints remain unchanged. However, note that $||w_{u_1} - w_{u_2}||_1\ne ||w'_{u_1} - w'_{u_2}||_1$ and $||w_{u_1} - w_{u_2}||_2\ne ||w'_{u_1} - w'_{u_2}||_2$; if we related $d$ to the $\ell_1$- or $\ell_2$-norm instead of the $\ell_\infty$-norm, the constraints would change despite the revised instance intuitively being symmetric to the original instance.
\end{itemize}

\subsection{Examples}\label{sec:examples}

We next present examples that illustrate our setting and the relevance of the various constraints. We start with one highlighting the fact that fairness constraints will not be naturally satisfied by profit-maximizing algorithms.

\begin{example}\label{ex:fairness-matters}
    \textbf{Importance of fairness constraints}. Consider a refugee resettlement setting with two types of agents $\{u_1, u_2\}$ representing refugees, and two facilities $\{v_1, v_2\}$ representing host cities. There are two types of resources $\{n_1, n_2\}$, corresponding to houses at each host city. There are $S = 100$ agents per batch and $T = 50$ batches. The probability of each agent type for each agent in each batch is $1/2$. The weight vectors, representing the probability of finding employment for each agent at each host city, is given by
    \[w_{u_1} = \begin{bmatrix}
        0.7 \\ 0.3
    \end{bmatrix}\qquad w_{u_2} = \begin{bmatrix}
        0.7-\epsilon \\ 0.3+\epsilon
    \end{bmatrix},\]
    for some small $\epsilon>0$. The resource consumption vectors are given by
    \[b_{u_1} = \begin{bmatrix}
        1 & 0 \\
        0 & 1
    \end{bmatrix}\qquad b_{u_2} = \begin{bmatrix}
        1 & 0 \\
        0 & 1
    \end{bmatrix}.\]
    Resource capacity is given by $B_{n_1} = B_{n_2} = ST/2 = 2500$. The offline unfair fluid linear program can be written as
    \begin{align*}
        \max \quad & \sum_{t=1}^{50} \sum_{s=1}^{100} \sum_{v \in \{v_1, v_2\}} w_{ts, v} x_{ts, v} \\
        \text{s.t.} \quad & \sum_{t=1}^{50} \sum_{s=1}^{100} x_{ts, v} \le 2500 \quad \forall v \in \{v_1, v_2\} \\
        & \sum_{v \in \{v_1, v_2\}} x_{ts, v} \le 1 \quad \forall t, s \\
        & x_{ts, v} \ge 0 \quad \forall t, s, v.
    \end{align*}
    Since agents have the same resource requirements, the unique optimal solution is to match agents of type $u_1$ to $v_1$, and match agents of type $u_2$ to $v_2$, namely
    \[x_{ts} = \begin{bmatrix}
            1 \\
            0
        \end{bmatrix}\text{ if } u_{ts} = u_1 \text{ and }\begin{bmatrix}
            0 \\
            1
        \end{bmatrix}\text{ if } u_{ts} = u_2,\]
    (with minor adjustments for the realized agents appearances, if there are not exactly 2500 of each type). Classical online allocation algorithms without consideration for fairness will also learn this rule. For example, classical online stochastic gradient descent (e.g. \citet{agrawal2014fast}) will learn the optimal dual variables, or shadow prices, for the dual LP
    \begin{align*}
        \min \quad & 2500 \lambda_{v_1} + 2500 \lambda_{v_2} + \sum_{t=1}^{50} \sum_{s=1}^{100} \mu_{ts} \\
        \text{s.t.} \quad & \lambda_v + \mu_{ts} \ge w_{ts, v} \quad \forall t, s, v \in \{v_1, v_2\} \\
        & \lambda_{v_1}, \lambda_{v_2}, \mu_{ts} \ge 0,
    \end{align*}
    where $\lambda_{v_1}, \lambda_{v_2}$, and $\mu_{ts}$ for $t\in [T], s\in [S]$ are the dual variables. The optimal solution will be approximately $\lambda_{v_1} = 0.4$, $\lambda_{v_2} = 0$, $\mu_{ts} = 0.3$ when $u_{ts} = u_1$, and $\mu_{ts} = 0.3+\epsilon$ when $u_{ts} = u_2$. This results in all agents of type $u_1$ being assigned to $v_1$, and all agents of type $u_2$ being assigned to $v_2$. We can also see, though, that this outcome is highly unfair for agents of type $u_2$: they have very similar weight and resource consumption characteristics to type $u_1$, but are completely prevented from matching with the more desirable host city $v_1$. Now consider the additional fairness constraints
    \[d(u_{ts_1}, u_{ts_2}) \ge \gamma (\mathbb E[w_{ts_1}^\top \tilde x_{ts_1}] - \mathbb E[w_{ts_2}^\top \tilde x_{ts_2}]) \quad \forall t \in [T], \forall s_1, s_2 \in [S],\]
    where $d(u_1, u_2) := ||w_{u_1} - w_{u_2}||_\infty + \underline{d}||b_{u_1} - b_{u_2}||_\infty = \epsilon$ and $\gamma = 1$. Then, the fair offline optimal solution would become
    \[x_{ts} = \begin{bmatrix}
            0.5 + \delta \\
            0.5 - \delta
        \end{bmatrix}\text{ if } u_{ts} = u_1 \text{ and }\begin{bmatrix}
            0.5 - \delta \\
            0.5 + \delta
        \end{bmatrix}\text{ if } u_{ts} = u_2,\]
    where $\delta = 5\epsilon / 4$. Essentially, each agent has approximately a 50\% chance of being assigned to each host city. This solution is a much more equitable outcome that also does not sacrifice significant optimality, with the difference in total welfare being bounded by an $O(\epsilon)$ factor.
\end{example}

From this example, we can see the importance of explicit fairness constraints in online resource allocation problems. We also believe that by guaranteeing provable fairness, a resource allocation algorithm is more likely to find acceptance among practitioners in various fields, especially where human considerations are paramount and headlong advances in uninterpretable artificial intelligence models have fomented an understandable distrust in automated systems. As we show in Section~\ref{sec:offline}, it is always possible to (possibly significantly) change the optimal unfair solution to a fair solution that guarantees a fraction of the total objective value that only depends linearly on $\gamma$.

We show in Appendix~\ref{app:examples} two additional examples showing the importance of Assumptions~\eqref{d_assumption_w} and \eqref{d_assumption_b} on the distance function. The examples demonstrate that without these assumptions, the fairness constraints can lead to highly unintuitive restrictions, or overly punish certain agents and drastically reduce the optimal objective function value.

\section{Bounding the Price of Fairness}\label{sec:offline}

In this section, our goal is to bound the loss in value in the offline fluid allocation problem from the addition of fairness constraints. That is, we want to give a lower bound on the value of $\gfair(\{u_{ts}\}_{t\in [T], s\in [S]})$ as a function of the value of $\unf(\{u_{ts}\}_{t\in [T], s\in [S]})$. To do this, we describe an algorithm that produces a constant approximation to the offline unfair problem, and simultaneously satisfies the fairness constraints. Thus, there exists a fluid solution that satisfies the fairness constraints and reduces the objective value only by a constant factor. This fact implies the following main result of this section. 

\begin{theorem}\label{thm:fairness}
For any $\gamma\ge 1/2$, the value of $\gfair$ is at least a $\Omega(1/\gamma)$ factor of the value of $\unf$.
\end{theorem}
\addtocounter{theorem}{1}

In the previous theorem, the $\Omega(\cdot)$ notation hides a linear dependence on $ 1/\underline{d},  1/\overline{b}$. For ease of exposition, we define $\Xi$ to be the set of pairs of indices $t,s$ for $t\in [T]$ and $s\in [S]$. Thus, $\Xi$ represents the set of all arriving agents. We can now present Algorithm~\ref{alg:water}, beginning with a high-level description.

\begin{itemize}
\item \textbf{Initialization:} Let $y^*\in [0,1]^{\Xi\times V}$ be the offline optimal matching solution for $\unf$. Let $x\in [0,1]^{\Xi\times V}$ be the current matching produced by the algorithm, initialized to $0$. Thus, for $\xi \in \Xi$, $x_\xi$ denotes row $\xi$ of $x$, which gives the fractional amount (i.e., the probability) of each facility being assigned to $\xi$. We will refer to the increase of $x_\xi$ along some vector $y^*_{\zeta}$ as \emph{filling along $y^*_{\zeta}$}. Throughout the algorithm, for each $\xi\in \Xi$, we let $a_\xi := w_\xi^\top x_\xi$ be the total value accumulated by agent $\xi$ so far.

\item \textbf{Iteration:} Each iteration is composed of the following steps:
\begin{itemize}
\item \emph{Sorting:} We sort the agents $\xi_1,\xi_2,\cdots$ in non-increasing value of $a_{\xi}$.
\item \emph{Selection of filling vector:}
For each $i$, we initialize $\zeta_i\gets\xi_i$. That is, by default $\xi_i$ fills along  $y^*_{\xi_i}$, the allocation it receives in the unfair optimal solution. Then, in order of $i=1,\dots, |\Xi|$, an agent $\xi_{i}$ sets $\zeta_i\gets\zeta_j$ for some $j<i$ that maximizes $w^\top_{\xi_i} y^*_{\xi_j}$ if any additional increase along the current vectors would violate the fairness constraint.
    \item \emph{Increase along filling vectors:} For each agent $\xi_i$, the total allocation $x_{\xi_i}$ is augmented by $\delta y^*_{\zeta_i}$ for the largest possible $\delta>0$, until a resource is fully consumed, a fairness constraint becomes binding, or the time runs out. We refer to $\delta$ as the amount of \emph{time} spent filling.
    \item \emph{Resources:} The total amount of resources consumed throughout the algorithm when agents fill along vector $y^*_{\xi_i}$ cannot exceed the total resources $b_{\xi_i}^\top y^*_{\xi_i}$ consumed by agent $\xi_i$ in the unfair optimal allocation. Once agents filling along a vector $y^*_{\xi_i}$ consume the available resources, we say that $\xi_i$ becomes \emph{closed} and is added to the set $\Xi'$. All vectors $y^*_{\xi}$ with $\xi\in \Xi'$ can no longer be filled by any agents.
       \item \emph{Renunciation of capacity:}
        If an agent $\xi_i$ consumes $\delta y^*_{\zeta_i}$ following an agent $\zeta_i\ne \xi_i$, it simultaneously ``renounces'' an equal amount of capacity $\delta y^*_{\xi_i}$ from its original allocation. This capacity counts towards the total amount of resources consumed filling along vector $y^*_{\zeta_i}$ (see previous bullet point), but does not contribute to the final allocation $x_\xi$ or value $a_{\xi}$.
 \end{itemize}
\item \textbf{Termination:} The sum of all values $\delta$ across each iteration of the algorithm performed so far is called \emph{time} and denoted by $\tau$. Note that $||x_{\xi}||_1\le \tau$ at any step of the algorithm. The algorithm terminates at time equals $1$, i.e., when $\tau=1$.
\end{itemize}

\begin{algorithm}
	\SetAlgoNoLine
    \DontPrintSemicolon
    \SetKwComment{Comment}{/* }{ */}
	\KwIn{Facilities $V$, agents $\Xi$, weights $w\in [0,1]^{\Xi\times V}$, distances $d:U\times U\to \mathbb R$, resource matrices $b_\xi\in \mathbb N^{V\times N}$, capacities $B_0\in \mathbb N^N$.}
	\KwOut{Matching $x$.}
    $B'_\xi\gets b_\xi^\top y^*_\xi$\label{alg:water:B'}\Comment*[r]{track resource consumption by agent}
    \For{$\xi\in \Xi$}{
        $x_\xi\gets 0\in [0,1]^{V}$\Comment*[r]{current assignment}
        $a_\xi\gets 0$\Comment*[r]{current value}
    }
	$\tau\gets 0$ \Comment*[r]{variable measuring time}
    $\Xi'\gets \emptyset$\Comment*[r]{the list of closed filling vectors is initialized to empty}
    \While{$\tau < 1$\label{alg:water:while}}{
        $\xi_1,\xi_2,\dots\gets \Xi$ sorted in non-increasing order of $a_\xi$\;
        $\zeta_i\gets \xi_i\quad\forall i\in [|\Xi|]$ \;
        \For{$i=1$ to $|\Xi|$\label{alg:water:for_vector}}{
            $\zeta_i\gets \textsc{SelectFill}(i, \Xi, \Xi', w, \zeta, a)$\Comment*[r]{each agent $\xi_i$ selects best eligible agent $\zeta_i$ to copy filling}
        }\label{alg:water:endfor_vector}
        $\delta\gets \textsc{StepSize}(\Xi,\zeta, w, a, b, B', d, y^*, \tau)$ \Comment*[r]{The largest step size that keeps $x$ fair and feasible is computed}
        $x_{\xi_i}\gets x_{\xi_i} + \delta y^*_{\zeta_i}\quad\forall i\in [|\Xi|]$\label{alg:water:allocated}\Comment*[r]{fill}
        $B'_{\xi_i}\gets B'_{\xi_i} - \delta(b_{\xi_i}+\sum_{j\ne i:\zeta_j = \xi_i}b_{\xi_j})^\top y^*_{\xi_i}$\label{alg:water:update_resources} \Comment*[r]{update resource consumption}
        $a_{\xi_i}\gets w_{\xi_i}^\top x_{\xi_i}\quad\forall i\in [|\Xi|]$\label{alg:water:allocated_weight}\Comment*[r]{update total values}
        $\Xi'\gets \Xi'\cup \{\xi_i\mid \exists k\in [N] \text{ s.t. } (B'_{\xi_i})_k = 0 \text{ and } ((b_{\xi_i}+\sum_{j\ne i:\zeta_j = \xi_i}b_{\xi_j})^\top y^*_{\xi_i})_k > 0\}$\label{alg:water:Xi'} \Comment*[r]{update list of closed filling vectors}
        $\tau\gets\tau + \delta$\;
    }
	\caption{Fair water-filling algorithm}
	\label{alg:water}
\end{algorithm}

\begin{algorithm}
	\SetAlgoNoLine
    \DontPrintSemicolon
    \SetKwComment{Comment}{/* }{ */}
	\KwIn{Index $i\in [|\Xi|]$, agents $\Xi$ sorted in non-increasing order of $a_\xi$, set $\Xi'$ of closed agents, fill indices $\zeta_i$, weights $w\in [0,1]^{\Xi\times V}$, current values $a_\xi$.}
	\KwOut{Fill index $\zeta_i$.}
    $J \gets \{j < i \mid a_{\xi_j} - a_{\xi_i} > d(u_{\xi_i}, u_{\xi_j}) \text{ or } (a_{\xi_j} - a_{\xi_i} = d(u_{\xi_i}, u_{\xi_j}) \text{ and } w_{\xi_j}^\top y^*_{\zeta_j}> w_{\xi_i}^\top y^*_{\zeta_i})\}$\label{alg:selectfill:J}\;
    \If{$\xi_i\not\in \Xi'$\label{alg:selectfill:if_capacity}}{
        $J\gets J\cup\{i\}$\label{alg:selectfill:J2}\;
    }
    \eIf{$J\ne\emptyset$}{
        $j^*\gets \arg\max_{j\in J}w_{\xi_i}^{\top}y^*_{\zeta_{j^*}}$\label{alg:selectfill:j*}\Comment*[r]{each agent tracks best eligible agent to follow}
        $\zeta_i\gets \zeta_{j^*}$\label{alg:selectfill:zeta}\;
    }{
        $\zeta_i\gets 0$\Comment*[r]{Define $y^*_0:= 0$}
    }
	\caption{\textsc{SelectFill}}
	\label{alg:selectfill}
\end{algorithm}

\begin{algorithm}
	\SetAlgoNoLine
    \DontPrintSemicolon
    \SetKwComment{Comment}{/* }{ */}
	\KwIn{Agents $\Xi$ sorted in non-increasing $a_{\xi}$, fill indices $\zeta$, weights $w\in [0,1]^{\Xi\times V}$, current values $a_{\xi}$, resource vectors $b$, current resources $B'$, distances $d:U\times U\to \mathbb R$, matching $y^*$, time $\tau$}
	\KwOut{Step size $\delta\in [0,1]$.}
    $\delta_1\gets$ max value s.t. $a_{\xi_j} + \delta_1 w_{\xi_j}^\top y^*_{\zeta_j} - a_{\xi_i} - \delta_1 w_{\xi_i}^\top y^*_{\zeta_i}\le d(u_{\xi_i}, u_{\xi_j})\ \forall i, j$ s.t. $a_{\xi_j} - a_{\xi_i}< d(u_{\xi_i}, u_{\xi_j})$\label{alg:stepsize:delta1}\Comment*[r]{Activation of fairness constraint}
    $\delta_2\gets$ max value s.t. $B'_{\xi_i} - \delta(b_{\xi_i}+\sum_{j\ne i:\zeta_j = \xi_i}b_{\xi_j})^\top y^*_{\xi_i} \ge 0$\label{alg:stepsize:delta2}\Comment*[r]{Activation of resource constraint}
    $\delta_3\gets 1-\tau$\label{alg:stepsize:delta3}\;
    $\delta\gets\min\{\delta_1,\delta_2,\delta_3\}$\;
	\caption{\textsc{StepSize}}
	\label{alg:stepsize}
\end{algorithm}

Note that the exact running time of Algorithm~\ref{alg:water} is not relevant; if we wanted to find the optimal $\gamma$-fair solution, we could simply solve the linear program directly. Algorithm~\ref{alg:water} is simply used as an analysis tool. We use the following helper results to prove Theorem~\ref{thm:fairness}. Proofs omitted from this section are deferred to Appendix~\ref{sec:app:offline}.

\begin{lemma}\label{lem:Alg:waterfill-terminates}
    Algorithm~\ref{alg:water} terminates in finite time.
\end{lemma}

The following three lemmas imply that the output of Algorithm~\ref{alg:water} is feasible for $\gfair(\{u_{ts}\}_{t\in [T], s\in [S]})$ with $\gamma=1/2$.

\begin{lemma}\label{lem:water_fsbl}
    The allocation $x$ returned by Algorithm~\ref{alg:water} satisfies the resource constraint. 
\end{lemma}

\begin{lemma}\label{lem:water_matching}
    The allocation $x$ returned by Algorithm~\ref{alg:water} satisfies the matching constraint.
\end{lemma}

\begin{lemma}\label{lem:water_fair}
    The allocation $x$ returned by Algorithm~\ref{alg:water} is 1/2-fair.
\end{lemma}

The last ingredient for the proof of Theorem~\ref{thm:fairness} is to show that Algorithm~\ref{alg:water} is approximately optimal, with respect to unfair optimal $y^*$.

\begin{lemma}\label{lem:water_optimal}
    If $\overline b = 1$, the allocation $x$ returned by Algorithm~\ref{alg:water} is $\min\{\frac{1}{2}\underline{d}, \frac{3}{8}\}$-optimal. If $\overline b\ge 2$, the allocation $x$ returned by Algorithm~\ref{alg:water} is $\min\{\frac{1}{2}\underline{d}, \frac{1}{2\overline b}\}$-optimal.
\end{lemma}

We can now prove Theorem~\ref{thm:fairness}.

\begin{proof}[Proof of Theorem~\ref{thm:fairness}]
    Fix any $\gamma>0$, and let $x$ be the allocation returned by Algorithm~\ref{alg:water}. If $\gamma \le 1/2$, then by Lemmas~\ref{lem:water_fsbl}, \ref{lem:water_matching}, \ref{lem:water_fair}, and \ref{lem:water_optimal}, $x$ is a fractional matching that satisfies the resource constraints, is $\gamma$-fair, and is $O(1)$-optimal. For any $\gamma > 1/2$, set $x^*= \frac{x}{2\gamma}$. Then, $x^* < x$, and so $x^*$ is also a fractional matching that satisfies the resource constraints. By Lemma~\ref{lem:water_fair}, $x$ satisfies
    \[d(u_{\xi_1}, u_{\xi_2}) \ge \frac{1}{2} (w_{\xi_1}^\top x_{\xi_1} - w_{\xi_2}^\top x_{\xi_2})\nonumber\]
    for each $\xi_1, \xi_2\in \Xi$. Since $x = 2\gamma x^*$, it follows that
    \[d(u_{\xi_1}, u_{\xi_2}) \ge \gamma (w_{\xi_1}^\top x^*_{\xi_1} - w_{\xi_2}^\top x^*_{\xi_2})\nonumber\]
    giving $\gamma$-fairness. Finally, by Lemma~\ref{lem:water_optimal}, $x$ is $O(1)$-optimal, which means $x^*$ is $O(1/\gamma)$-optimal.
\end{proof}

\section{Online Fair Algorithm}\label{sec:algorithm}

  In this section, we present the online fair resource allocation algorithm. The first step is to rewrite the problem in a more convenient form. Note that, due to the fairness constraints, if two agents of the same type arrive in the same batch, then any fair algorithm must give each agent the same (fractional) allocation, since their distance is equal to $0$. So, in each time period an algorithm can simply decide how to allocate to each type, instead of each one of $S$ agents. Furthermore, the space of feasible allocations over the types is then bounded by the fairness constraints.

Let $\mathcal X\subseteq [0,1]^{U\times V}$ be the region given by
$$\sum_{v\in V} (X^\top e_u)_v\le 1 \qquad\forall u\in U$$
$$X\ge 0$$
That is, $\mathcal X$ is the space of all allocations of types to facilities. Let $\Pi$ be the set of sequences of $S$ elements from $U$, with replacement. We can therefore think of $\pi \in \Pi$ as the realization of one batch of arrivals. For each $\pi\in \Pi$, let $\mathcal X_\pi\subseteq \mathcal X$ be the region obtained by restricting $\mathcal X$ via the constraints
\[\gamma (w_{\pi(i)}^\top X^\top e_{\pi(i)} - w_{\pi(j)}^\top X^\top e_{\pi(j)}) - d(u_{\pi(i)},u_{\pi(j)})\le 0\qquad \forall i,j\in [S]. \nonumber\]
The previous constraints restrict the allocations in $\mathcal X$ to those that satisfy the fairness constraint for the batch corresponding to $\pi$.
Similarly, at time period $t\le T$, let $\mathcal X_t\subseteq \mathcal X$ be the region obtained by restricting $\mathcal X$ via the constraints
\[\gamma (w_{ti}^\top X^\top e_{ti} - w_{tj}^\top X^\top e_{tj}) - d(u_{ti},u_{tj})\le 0\qquad \forall i,j\in [S], \nonumber\]
where recall that we abbreviate $(w_{ts}, b_{ts}, e_{ts}):= (w_{u_{ts}}, b_{u_{ts}}, e_{u_{ts}})$ for $t \in [T]$ and $s \in [S]$.
With these new objects, the offline fair fluid problem $\flu(F)$ for a probability distribution $F$ can be written as:
\begin{align}
    \flu(F) = \mathbb E_{F}\begin{bmatrix}
        \max_{X_t\in \mathcal X_t} & \sum_{t=1}^T\sum_{s=1}^Sw_{ts}^\top X_{t}^\top e_{ts} \\
        \text{s.t.} & \sum_{t=1}^T\sum_{s=1}^S b_{ts}^\top X_t^\top  e_{ts}\le TS\rho
    \end{bmatrix} \label{eq:fair_primal}
\end{align}
where we recall $\rho:= B_0/(TS)$. For $c\in \mathbb R^{U\times V}$, define the conjugates
\[w^*_{\pi}(c) = \max_{X\in \mathcal X_\pi}\left\{\sum_{s=1}^S w_{\pi(s)}^\top X^\top e_{\pi(s)} - c_{\pi(s)}^\top X^\top e_{\pi(s)}\right\}, \nonumber\]
\[w^*_{t}(c) = \max_{X\in \mathcal X_t}\left\{\sum_{s=1}^S w_{ts}^\top X^\top e_{ts} - c_{ts}^\top X^\top e_{ts}\right\}. \nonumber\]
Let $p_{\pi}$ be the probability of $\pi$ under $F$. For $\mu\in \mathbb R^{N}$ define the offline fair dual problem $D(\mu)$ by
\[D(\mu) = \sum_{\pi\in \Pi} p_{\pi}w^*_{\pi}(\mathbf{b}\mu) + S\mu^\top\rho, \nonumber\]
where $(\mathbf{b}\mu)_u:= b_u\mu$. We can show that the dual problem bounds the primal problem.

\begin{lemma}\label{lem:dual}
    For all $\mu\ge 0$, $\flu(F)\le TD(\mu)$.
\end{lemma}

We defer the proof of Lemma~\ref{lem:dual} to Appendix~\ref{sec:app:algorithm}. In order to describe the fair online allocation algorithm, we first define the Bregman distance for a convex function. 

\begin{definition}[Bregman distance]
    Let $h: \mathbb{R}^N \rightarrow \mathbb{R}$ be any convex (reference) function. The Bregman distance is given by $V_h(x,y) = h(x) - h(y) - \nabla h(y)^\top(x-y)$.
\end{definition}

We can now present Algorithm~\ref{alg:inbatch_lazy}, the fair online allocation algorithm. The algorithm and its analysis are based on the online dual mirror descent algorithm by \citet{balseiro2020dual}. Online mirror descent, a generalization of online gradient descent, is a well-known algorithm in online convex optimization (see, e.g., \citet{srebro2011universality}). In our algorithm, we apply online mirror descent to the dual problem, instead of the primal problem. In principle, if we knew the optimal offline dual variables in advance, we could separate the online primal problem across time periods and compute optimal online allocations. By estimating the optimal dual variables via online dual mirror descent, we can approximate this approach. The main difference from \citet{balseiro2020dual} is that agents arrive in batches, and allocations must be $\gamma$-fair within each batch. We thus adapt the online dual mirror descent algorithm to make decisions for an entire batch at a time, which allows us to enforce the fairness constraints while maintaining low regret.

As it is customary, in Algorithm~\ref{alg:inbatch_lazy}, we let $\mu_0$ be any starting solution, e.g., the zero vector. We moreover define $\mu^{max}\in\mathbb R^{N}$ such that $\mu^{max}_{k} = \frac{\overline w}{\rho_{k}}+1$, and $\overline \mu := ||\mu^{max}||_\infty$. Thus, $\mu^{max}$ is a coordinate-wise upper bound to any feasible dual solution. 

\begin{algorithm}[H]
	\SetAlgoNoLine
    \DontPrintSemicolon
    \SetKwComment{Comment}{/* }{ */}
	\KwIn{Initial dual solution $\mu_0\le\mu^{max}$, number of batches $T$, batch size $S$, available resources $B_0 = TS\rho$, reference function $h$, step size $\eta$, fairness constant $\gamma$, distance metric $d$, agent types $U$.}
	\KwOut{Online matching $\tilde x$.}
    \For{$t = 1,\dots, T$}{
        Receive $(w_{ts}, b_{ts}, e_{ts})\sim F$ for all $s\in [S]$\;
        Compute action:
        \[X^*_{t} \gets \arg\max_{X\in \mathcal X_t} \sum_{s = 1}^{S} w_{ts}^\top X^\top e_{ts} - (b_{ts}\mu_{t-1})^\top X^\top e_{ts}+ S\mu^\top \rho  \nonumber\]\;
        Sample $\tilde x_{ts} \sim X^{*\top}_t e_{ts}$ randomly \;
        Update resources:           
        \[x_{ts} \gets \begin{cases}
            \tilde x_{ts} & \text{if } \sum_{s=1}^{S}b_{ts}^\top\tilde x_{ts} \le B_{t-1} \\
            0 & \text{o.w.}
        \end{cases} \nonumber\]
        \[B_t\gets B_{t-1} - \sum_{s = 1}^{S}b_{ts}^\top x_{ts} \nonumber\]\;
        Update dual variables by stochastic mirror descent:
        \[\tilde g_t \gets\sum_{s=1}^{S}(\rho-b_{ts}^\top X_t^{*\top} e_{ts}) \nonumber\]
        \[\mu_t \gets \arg\min_{\mu\ge 0} \tilde g_t^\top \mu + \frac{1}{\eta} V_h(\mu, \mu_{t-1}) \nonumber\]\;   
    }
	\caption{Online fair resource allocation algorithm}
	\label{alg:inbatch_lazy}
\end{algorithm}

\smallskip 

Since the randomized allocations $X^*_t$ computed in each iteration are within $\mathcal X_t$, by definition they must satisfy the fairness constraint. Note that the fairness constraints for the online problem are applied to the randomized allocations and not the realizations $\tilde x_{ts}\sim X_{t}^{*\top}e_{ts}$. Furthermore, if a realization violates the resource constraints, the allocation is dropped and the agents go unmatched, in order to satisfy both the fairness and resource constraints and to simplify analysis. In practice, more lenient policies can be used. Since at each time $t \in [T]$, $X^*_t \in {\mathcal X}_t$, the following holds.

\begin{corollary}
    Within each batch $t$, the ex-ante allocations $X^*_t$ produced by Algorithm~\ref{alg:inbatch_lazy} are $\gamma$-fair (see~\eqref{eq:lipsch-fair}).
\end{corollary}

We make the following assumptions about the reference function $h$:
\begin{itemize}
    \item $h$ is coordinate-wise separable, so that $h(\mu) = \sum_{j=1}^N h_j(\mu_j)$, where $h_j$ is a convex univariate function.
    \item $h$ is $\sigma_1$-strongly convex in $\ell_1$-norm for some constant $\sigma_1>0$, so that $h(\mu_1)\ge h(\mu_2) + \nabla h(\mu_2)^\top (\mu_1 - \mu_2) + \frac{\sigma_1}{2}||\mu_1 - \mu_2||_1^2$ for any $\mu_1, \mu_2$. 
    \item $h$ is $\sigma_2$-strongly convex in $\ell_2$-norm for some constant $\sigma_2>0$, so that $h(\mu_1)\ge h(\mu_2) + \nabla h(\mu_2)^\top (\mu_1 - \mu_2) + \frac{\sigma_2}{2}||\mu_1 - \mu_2||_2^2$ for any $\mu_1, \mu_2$.
\end{itemize}

These assumptions are standard in many algorithms based on mirror descent, see, e.g.,~\cite{balseiro2020dual}. An example of a reference function satisfying these assumptions is the squared Euclidean norm, $h(\mu) = \frac{1}{2}||\mu||_2^2$.

\section{Bounding the Price of Learning}\label{sec:online}

The main result of this section is Theorem~\ref{thm:regret}, which shows that our algorithm achieves $O(\sqrt{T})$ regret in expectation. In particular, we show the following result. Let $\theta$ denote the random variable governing the realization of $\tilde x_{ts}$ in Algorithm~\ref{alg:inbatch_lazy}.

\begin{manualthm}[2]\makeatletter\edef\@currentlabel{2}\makeatother\label{thm:regret}
    Consider Algorithm~\ref{alg:inbatch_lazy} with step size $\eta\le \frac{\sigma_2}{S\overline b}$ and $\mu_0\le\mu^{max}$. Then, the expected regret is bounded by 
    \begin{align*}
        \mathbb E_{\theta}[Regret(A\mid  F)]\le \frac{2 \eta S^2(\overline b^2 + \overline \rho^2)}{\sigma_1}T + \frac{1}{\eta}V_h(0, \mu_0) + \frac{\overline w}{\eta\underline \rho}||\nabla h(\mu^{max}) - \nabla h(\mu_0)||_\infty + \frac{S\overline w \overline b}{\underline \rho}.
    \end{align*}
    When we choose $\eta\in O(1/\sqrt{T})$, then $\mathbb E_{\theta}[Regret(A\mid  F)]\in O(\sqrt{T})$.
\end{manualthm}

Note that the $O(\sqrt{T})$ term also hides a dependence on the batch size $S$, which we assume to be constant. Due to the fairness constraints for each batch, if the constraints are sufficiently tight, any algorithm may be forced to give each agent in a batch essentially identical allocations. In this case, the agents in a batch act like a single large agent, and thus regret scales linearly in $S$. In general, though, selecting $\eta\in O(1/S)$ will still minimize the dependence on $S$.

To prove Theorem~\ref{thm:regret}, we define the stopping time $\tau_A$ as the first time we (nearly) run out of a resource.

\begin{definition}
    The stopping time $\tau_A$ of Algorithm~\ref{alg:inbatch_lazy} is the first time $\tau \leq T$ such that there exists a resource $j$ with
    \[\sum_{t = 1}^{\tau}\sum_{s=1}^S(b_{ts}^\top \tilde x_{ts})_j + S\overline b \ge \rho_j ST, \nonumber\]
and $\tau_A:=T$ if no such $\tau$ exists. Also, define
    \[\tilde \mu_{\tau_A}:= \frac{\sum_{t=1}^{\tau_A}\mu_t}{\tau_A}. \nonumber\]
\end{definition}

The first main intermediate result on our way to Theorem~\ref{thm:regret} shows that, in expectation, the stopping time $\tau_A$ is close to the final batch $T$. 

\begin{lemma}\label{lem:stopping_time}
    Consider Algorithm~\ref{alg:inbatch_lazy} with step size $\eta \le \frac{\sigma_2}{S\overline b}$. Then, $\mu_t\le \mu^{max}$ for all $t\leq \tau_A$. Furthermore, we have
    \[\mathbb E_\theta[S(T - \tau_A)] \le \frac{1}{\eta\underline\rho}||\nabla h(\mu^{max}) - \nabla h(\mu_0)||_\infty + \frac{S\overline b}{\underline\rho}. \nonumber\]
\end{lemma}

Proofs omitted from this section are deferred to Appendix~\ref{sec:app:online}. The next intermediate step for proving Theorem~\ref{thm:regret} is a helper result that bounds the total welfare achieved by Algorithm~\ref{alg:inbatch_lazy} before the stopping time $\tau_A$.

\begin{lemma}\label{lem:part2} 
    The total welfare before $\tau_A$ is bounded by
    \[\mathbb E_F\left[\tau_A D(\tilde \mu_{\tau_A}) - \sum_{t=1}^{\tau_A}\sum_{s=1}^Sw_{ts}^\top \tilde x_{ts}\right] \le \frac{2\eta S^2(\overline b^2 + \overline \rho^2)}{\sigma_1}\mathbb E_F[\tau_A] + \frac{1}{\eta}V_h(0, \mu_0). \nonumber\]
\end{lemma}

Finally, we can use Lemma~\ref{lem:stopping_time} and Lemma~\ref{lem:part2} to prove Theorem~\ref{thm:regret}.

\begin{proof}[Proof of Theorem~\ref{thm:regret}]
    For any $F\in\mathcal F$, we know by Lemma~\ref{lem:dual} and the fact that $FLU\le TS\overline w$ that
    \begin{align*}
        FLU(F) &= \frac{\mathbb E_\theta[\tau_A]}{T}FLU(F) + \frac{\mathbb E_\theta[T-\tau_A]}{T}FLU(F) \\
        &\le \mathbb E_\theta\left[\tau_A D(\tilde \mu_{\tau_A}) + (T - \tau_A)S\overline w\right]
    \end{align*}
    So, we can write the expected regret as
    \begin{align*}
        \mathbb E_\theta\left[Regret(A\mid F)\right] &= FLU(F) - \mathbb E_{\theta}[R(A\mid F)] \\
        &\le \mathbb E_F\left[\mathbb E_\theta\left[\tau_A D(\tilde \mu_{\tau_A}) + (T - \tau_A) S\overline w - \sum_{t=1}^{\tau_A} \sum_{s=1}^S w_{ts}^\top \tilde x_{ts}\right]\right]
     \end{align*}
    By Lemma~\ref{lem:stopping_time} and Lemma~\ref{lem:part2}, for some constant $c>0$, we have
    \begin{align*}
         \mathbb E_\theta\left[Regret(A\mid F)\right]&\le \frac{2 S^2(\overline b^2 + \overline \rho^2)}{\sigma_1}\eta\mathbb E_F[\tau_A] + \frac{1}{\eta}V_h(0,\mu_0) + \frac{\overline w}{\eta\underline \rho}||\nabla h(\mu^{max}) - \nabla h(\mu_0)||_\infty + \frac{S\overline w \overline b}{\underline \rho} \\
        &\le \frac{2 S^2(\overline b^2 + \overline \rho^2)}{\sigma_1}\eta T + \frac{1}{\eta}V_h(0,\mu_0) + \frac{\overline w}{\eta\underline \rho}||\nabla h(\mu^{max}) - \nabla h(\mu_0)||_\infty + \frac{S\overline w \overline b}{\underline \rho}
    \end{align*}
    as desired.
\end{proof}

\section{Experiments}\label{sec:experiments}

In this section, we evaluate the performance of Algorithm~\ref{alg:inbatch_lazy} using real-world data from the Refugee Economies Programme. Our primary goal is to empirically validate our theoretical bounds and assess the ``price of fairness''—the trade-off between maximizing refugee welfare and satisfying Lipschitz fairness constraints.

\subsection{Dataset and Simulation Environment}

We utilize the publicly available Refugee Economies Cross-country Dataset, prepared by the Refugee Studies Centre at the University of Oxford \citep{betts2024economic}. This dataset aggregates household survey data from six distinct data collection exercises conducted between 2016 and 2021 across three countries: Kenya, Uganda, and Ethiopia.

The simulation environment is constructed using the following components derived from the data:
\begin{enumerate}
    \item \textbf{Facilities}: We define the set of facilities $V$ as the host locations, given by the six survey sites provided in the dataset, namely
    \begin{itemize}
        \item Kakuma refugee camp (Kenya)
        \item Nairobi (Kenya)
        \item Nakivale refugee camp (Uganda)
        \item Kampala (Uganda)
        \item Dollo Ado (Ethiopia)
        \item Addis Ababa (Ethiopia).
    \end{itemize}
    \item \textbf{Refugee Types and Arrivals}: The set of agents is given by the 3674 refugee families in the dataset. Each family is defined as a unique type in $U$. For online simulation, we group cases into $T = 50$ batches.
    \item \textbf{Weights}: The standard objective in the refugee resettlement literature is to maximize the probability that any individual in a family finds employment within a year of resettlement (see, e.g. \citet{bansak2018improving} and \citet{ahani2021placement}). We encounter the typical challenges when using employment probabilities as the primary objective metric. The weights $(w_{u})_{v}$ are unobserved in practice, and even on our historical dataset, the realized employment outcome for each family is only observed at the host location the family was assigned to in reality. We take the standard approach to ameliorating these issues: as in \citet{bansak2018improving} and \citet{ahani2021placement}, we train a machine learning model on the dataset to estimate the employment probability for each family at each host location.

    \item \textbf{Resources and Capacities}: We model a single resource type for each location, which reflects the capacity of each location to house new arrivals. For the purpose of this experiment, we define capacities proportional to the sample sizes of each location in the dataset, ensuring the simulation reflects the relative scale of the camps and urban areas.
\end{enumerate}

\subsection{Experimental Design and Baselines}

We compare the performance of our proposed fair algorithm against a standard baseline online resource allocation algorithm that prioritizes efficiency without regard for fairness.

\begin{itemize}
    \item We run Algorithm~\ref{alg:inbatch_lazy}, which incorporates the Lipschitz fairness constraints. We vary the fairness parameter $\gamma$ to observe how strictly enforcing similar outcomes for similar individuals impacts the total welfare. 
    \item We compare our approach against the standard Dual Mirror Descent algorithm for online allocation problems, as described by \citet{balseiro2020dual}. This baseline optimizes for total expected employment but does not enforce the local fairness constraints within batches.
\end{itemize}

The experiments simulate the arrival of refugee batches over the finite horizon $T$. We consider the following metrics:
\begin{itemize}
    \item \textbf{Total Welfare}: The cumulative expected employment achieved by the allocation.
    \item \textbf{Fairness}: The realized fairness parameter $\gamma$, particularly of the unfair benchmarks.
\end{itemize}

\subsection{Experimental Results}

We evaluate the offline fair fluid solutions to $\gfair$, and the online fair solution produced by Algorithm~\ref{alg:inbatch_lazy}, with various fairness coefficients of $\gamma$. As benchmarks, we compare to the offline unfair fluid solution for $\unf$, and the online unfair solution produced by the online resource allocation algorithm of \citet{balseiro2020dual}.

Results are presented in Table~\ref{table:offline_comp} and visually in Figure~\ref{fig:offline_comp}. Small values of $\gamma$ have very minimal impact on the offline fluid value, with the fluid $1$-fair solution still producing over 98\% of the value of the fluid unfair solution. Even as $\gamma$ increases to 4, the offline fluid $4$-fair solution still produces 73\% of the value of the fluid unfair solution. For the online setting, Algorithm~\ref{alg:inbatch_lazy} also performs well in comparison to the offline solution despite being a fair online algorithm, consistently achieving at least 90\% of the offline fluid fair value.

\begin{table}[h]
\centering
\begin{tabular}{|l|c|c|c|c|c|}
\hline
 & Benchmark & $\gamma=0.5$ & $\gamma=1$ & $\gamma=2$ & $\gamma=4$ \\ \hline
\textbf{Offline Value} & 2679.40 & 2673.35 & 2647.13 & 2450.49 & 1973.93 \\ \hline
\textbf{Offline \% of Baseline} & - & 99\% & 98\% & 91\% & 73\% \\ \hline
\textbf{Online Value}& 2583.86 & 2395.51 & 2375.87 & 2233.60 & 1945.88 \\ \hline
\textbf{Online \% of Offline} & 96\% & 90\% & 90\% & 91\% & 99\% \\ \hline
\end{tabular}
\caption{Offline optimal and online objective values for unfair baseline $\unf$ and fair solutions to $\gfair$ with various choices of $\gamma$. \textit{Offline \% of Baseline} is computed as [Offline $\gamma$-Fair Value]/[Offline Unfair Benchmark]$\times 100\%$. \textit{Online \% of Offline} is computed as [Online $\gamma$-Fair Value]/[Offline $\gamma$-Fair Value]$\times 100\%$.} \label{table:offline_comp}
\end{table}

\begin{figure}
    \centering
    \includegraphics[width=0.6\linewidth]{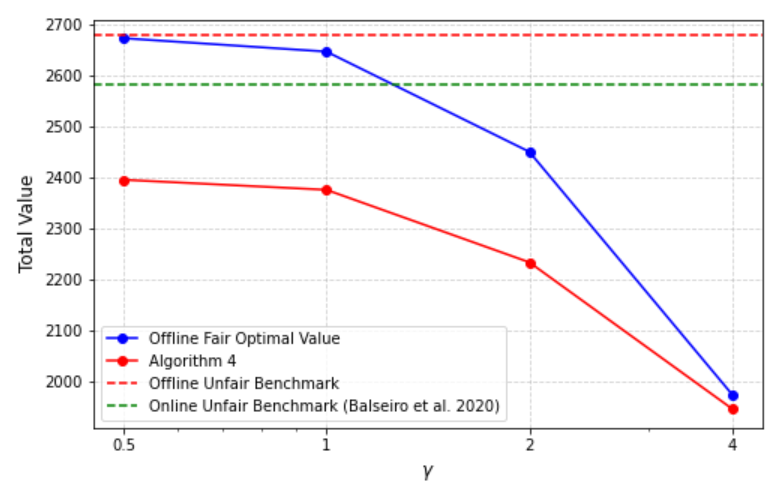}
    \caption{Total objective values for offline unfair benchmark, online unfair fluid benchmark \citep{balseiro2020dual}, offline fair fluid, and Algorithm~\ref{alg:inbatch_lazy}.}
    \label{fig:offline_comp}
\end{figure}

In Figure~\ref{fig:unfair_fairness}, we also examine the fairness properties of the offline unfair fluid solution $\unf$ and the online unfair benchmark algorithm of \cite{balseiro2020dual}. For each solution, we take each pair of agents with types $u_{ts_1}, u_{ts_2}$ within each batch $t$, and measure the realized value of $\gamma^*$ such that $d(u_1, u_2) = \gamma^*|\mathbb E[w_1^\top \tilde x_1] - \mathbb E[w_2^\top \tilde x_2]|$. We then present the distributions of realized $\gamma^*$ as a box-and-whiskers plot, marking each quartile of each distribution. Note that in Figure~\ref{fig:unfair_fairness} we omit rightward outliers. The key point is to examine the left or bottom quartiles of the fairness distributions. As we can see, the solution of~\citet{balseiro2020dual} has a significant proportion of agents who are treated unfairly, with realized $\gamma^*$-fairness essentially approaching 0 at the bottom end. In each unfair benchmark, around $1/4$ of the realized $\gamma^*$ values are less than 1, and around $1/2$ are less than 2. Instead, Algorithm~\ref{alg:inbatch_lazy} forces a lower bound on the fairness of the outcome, based on the chosen input value of $\gamma$. Figure~\ref{fig:welfare_plots} also shows the distribution of ex-ante (expected) welfare outcomes for Algorithm~\ref{alg:inbatch_lazy} vs. the unfair online benchmark. That is, we show the expected employment objective for each agent. While the distributions are similar, with similar mean and median welfare, the unfair algorithm has a much longer and heavier left tail where a significant proportion of cases achieve essentially 0 welfare. In the fair algorithm, the worst outcomes are avoided.

Finally, we examine the regret scaling of Algorithm~\ref{alg:inbatch_lazy} vs. the horizon $T$. For this part, we ran 100 trials, each selecting a horizon length $T$ uniformly at random on $[1,1000]$. Then, we bootstrap a dataset of length $T$ with batches of size $S=10$ by randomly sampling cases from the household survey data, and compared the value of Algorithm~\ref{alg:inbatch_lazy} with the fair fluid value $\gfair$. In Figure~\ref{fig:regret_vs_T}, which shows the regret vs. the time horizon $T$, we can see that the regret scales as $O(\sqrt{T})$. In Figure~\ref{fig:relreward_vs_T}, we plot the relative reward of Algorithm~\ref{alg:inbatch_lazy} compared to the offline fair fluid value, i.e. the online value divided by the offline fair fluid value. As $T$ increases to 1000, the relative reward approaches 1, showing that the relative price of learning essentially goes to 0.

\begin{figure}
    \centering
    \includegraphics[width=0.9\linewidth]{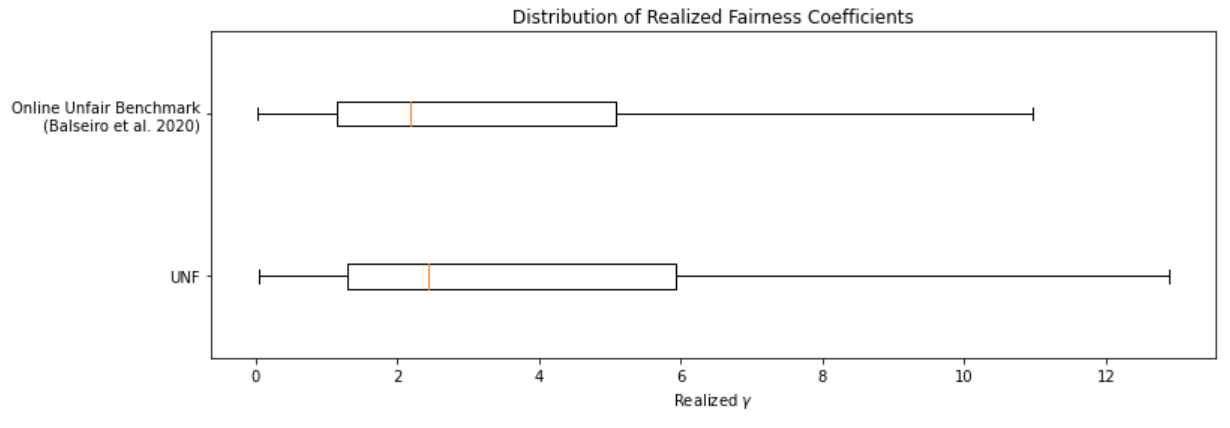}
    \caption{Distribution of realized pairwise fairness coefficients $\gamma$ for unfair benchmarks.}
    \label{fig:unfair_fairness}
\end{figure}

\begin{figure}
    \centering
    \includegraphics[width=0.5\linewidth]{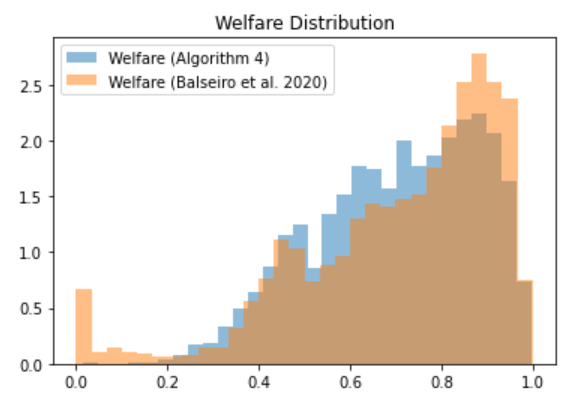}
    \caption{Distribution of ex-ante welfare, or expected employment for each agent, obtained by Algorithm~\ref{alg:inbatch_lazy} vs. unfair online benchmark.}
    \label{fig:welfare_plots}
\end{figure}

\begin{figure}[htbp]
     \centering
     \begin{subfigure}[b]{0.45\textwidth}
         \centering
         \includegraphics[width=\textwidth]{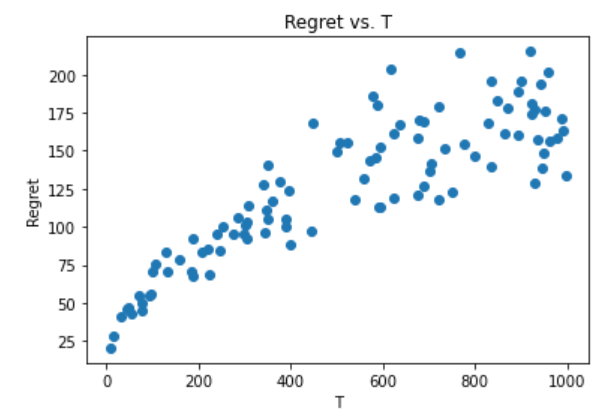}
         \caption{Regret vs. time horizon $T$}
         \label{fig:regret_vs_T}
     \end{subfigure}
     \hfill
     \begin{subfigure}[b]{0.45\textwidth}
         \centering
         \includegraphics[width=\textwidth]{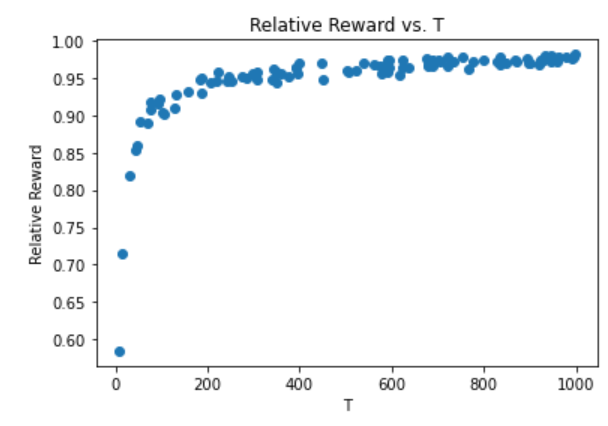}
         \caption{Proportion of offline fluid value vs. $T$}
         \label{fig:relreward_vs_T}
     \end{subfigure}
     \caption{Performance metrics for Algorithm~\ref{alg:inbatch_lazy} across different time horizons.}
     \label{fig:combined_metrics}
\end{figure}

\section*{Acknowledgements} Christopher En and Yuri Faenza acknowledge support from the NSF Grant 2046146. Christopher En acknowledges the support of the Columbia University Presidential Fellowship and Cheung-Kong Innovation Doctoral Fellowship.

\newpage
\nocite{*}
\bibliographystyle{apalike}
\bibliography{main}

\newpage
\appendix

\section{Omitted Examples}~\label{app:examples}

In this section, we have two additional examples that show the importance of our assumptions on $d$. The first shows that without Assumption~\ref{d_assumption_w}, even solutions that are intuitively fair become restricted.

\begin{example}\label{ex:assumption-2}
    \textbf{Importance of Assumption~\ref{d_assumption_w} on $d$}. Consider a setting with two types of agents $\{u_1, u_2\}$ and two facilities $\{v_1, v_2\}$. There are two types of resources $\{n_1, n_2\}$, representing a finite resource associated with each facility. There are $S=100$ agents per batch and $T=1$ batch. The probability of each agent type for each agent is $1/2$. The weight and resource consumption vectors are given by
    \[w_{u_1} = \begin{bmatrix}
        1 \\ 1
    \end{bmatrix}\qquad w_{u_2} = \begin{bmatrix}
        0.5 \\ 0.5
    \end{bmatrix}\]
    \[b_{u_1} = \begin{bmatrix}
        1 & 0 \\
        0 & 1
    \end{bmatrix}\qquad b_{u_2} = \begin{bmatrix}
        1 & 0 \\
        0 & 1 
    \end{bmatrix}.\]
    Define the fairness constraint as
    \[d(u_{ts_1}, u_{ts_2}) \ge \gamma (\mathbb E[w_{ts_1}^\top \tilde x_{ts_1}] - \mathbb E[w_{ts_2}^\top \tilde x_{ts_2}]) \quad \forall t \in [T], \forall s_1, s_2 \in [S],\]
    where $d(u_{ts_1}, u_{ts_2}):= \frac{1}{2}||w_{ts_1} - w_{ts_2}||_\infty + \underline d||b_{ts_1} - b_{ts_2}||_\infty$ and $\gamma = 1$. Note that $d$ violates the second assumption for distance functions, with the additional $1/2$ factor. Consider the assignment 
    \[x_{ts} = \begin{bmatrix}
            0.5 \\
            0.5
        \end{bmatrix},\]
    for all $t,s$. That is, every agent has an equal chance of being assigned to either facility. This assignment is not $1$-fair for this distance function, despite intuitively being a highly fair assignment. If instead we adjust $d$ to satisfy the assumptions on distance functions, setting $d(u_{ts_1}, u_{ts_2}):= ||w_{ts_1} - w_{ts_2}||_\infty + ||b_{ts_1} - b_{ts_2}||_\infty$, we see that $x$ now satisfies the $1$-fairness constraint. In general, the assumption that $d(u_i, u_j)\ge ||w_i - w_j||_\infty$ ensures that any assignment where each agent receives the same (fractional) allocation is 1-fair.
    
\end{example}

The next example shows that without Assumption~\ref{d_assumption_b}, the fairness constraints may lead to a significant loss in the objective, and unreasonably punish some agents.

\begin{example}\label{ex:assumption-3}
    \textbf{Importance of Assumption~\ref{d_assumption_b} on $d$}.
    Consider a setting with two types of agents $\{u_1, u_2\}$ and two facilities $\{v_1, v_2\}$. There are three types of resources $\{n_1, n_2, n_3\}$. There are $S=100$ agents per batch and $T=1$ batch. The probability of each agent type for each agent is $1/2$. The weight and resource consumption vectors are given by
    \[w_{u_1} = \begin{bmatrix}
        1 \\ 0.5
    \end{bmatrix}\qquad w_{u_2} = \begin{bmatrix}
        1 \\ 0.5
    \end{bmatrix},\]
    \[b_{u_1} = \begin{bmatrix}
        1 & 1 & 0  \\
        0 & 0 & 1 
    \end{bmatrix}\qquad b_{u_2} = \begin{bmatrix}
        1 & 0 & 0  \\
        0 & 0 & 1 
    \end{bmatrix}.\]
    Resource capacity is given by $B_{n_1} = B_{n_3} ST/2 = 50$, and $B_{n_2} = 10$. Define the fairness constraint as
    \[d(u_{ts_1}, u_{ts_2}) \ge \gamma (\mathbb E[w_{ts_1}^\top \tilde x_{ts_1}] - \mathbb E[w_{ts_2}^\top \tilde x_{ts_2}]) \quad \forall t \in [T], \forall s_1, s_2 \in [S],\]
    where $d(u_{ts_1}, u_{ts_2}):= ||w_{ts_1} - w_{ts_2}||_\infty$. Note that $d$ violates the third assumption for distance functions. We see that $d(u_i, u_j) = 0$ for all $i, j$. The welfare-optimal assignment is given by
    \[x_{ts} = \begin{bmatrix}
            0.2 \\
            0.8
        \end{bmatrix}\text{ if } u_{ts} = u_1 \text{ and }\begin{bmatrix}
            0.5 \\
            0.2
        \end{bmatrix}\text{ if } u_{ts} = u_2.\]
    Under this assignment, each agent has expected utility $w_{ts}^\top x_{ts} = 0.6$. Furthermore, the total resource consumption is
    \[\sum_{t\in [T], s\in [S]}b_{ts}^\top x_{ts} = \begin{bmatrix}
        35 \\
        10 \\
        50
    \end{bmatrix}\]
    This assignment is intuitively suboptimal, as capacity of resource $n_1$ is wasted while agents of type $u_2$ go under-allocated. Furthermore, agents of type $u_2$ are punished for the lack of a resource they do not consume ($n_2$). This scenario arises because the two agent types, which have different resource requirements, are considered ``equivalent'' by the distance function. If we instead define $d(u_{ts_1}, u_{ts_2}):= ||w_{ts_1} - w_{ts_2}||_\infty + \underline d||b_{ts_1} - b_{ts_2}||_\infty$ with $\underline d = 0.3$, satisfying the assumptions on the distance function, the most efficient $1$-fair allocation is given by
    \[x_{ts} = \begin{bmatrix}
            0.2 \\
            0.8
        \end{bmatrix}\text{ if } u_{ts} = u_1 \text{ and }\begin{bmatrix}
            0.8 \\
            0.2
        \end{bmatrix}\text{ if } u_{ts} = u_2,\]
    which is a solution that wastes fewer resources, and does not excessively punish agents of type $u_2$.
\end{example}

\section{Proofs Omitted from Section~\ref{sec:offline}}\label{sec:app:offline}

First, we prove Lemma~\ref{lem:Alg:waterfill-terminates}. For $\tau \in [0,1]$, we let $a^\tau$ be the vector $a$ at time $\tau$ of the algorithm.

\begin{proof}[Proof of Lemma~\ref{lem:Alg:waterfill-terminates}]
    In each iteration of the while loop, one of several possible events may occur, relating to the selection of $\delta$ in subroutine \textsc{StepSize}.

    \textbf{Case 1}: $\delta\gets \delta_3 = 1-\tau$. This may occur exactly once.

    \textbf{Case 2}: $\delta\gets \delta_2$. That is, $\delta$ is set to the maximum value s.t.~$B'_{\xi_i} - \delta(b_{\xi_i}+\sum_{j\ne i:\zeta_j = \xi_i}b_{\xi_j})^\top y^*_{\xi_i} \ge 0$. That is, a resource constraint activates. Then, in line~\ref{alg:water:Xi'} of Algorithm~\ref{alg:water}, the vector $y^*_{\xi_i}$ is closed and $\xi_i$ is added to $\Xi'$. In all future iterations of subroutine \textsc{SelectFill}, $\xi_i$ can never be selected. Thus, this case can occur at most $|\Xi|$ times.

    \textbf{Case 3}: $\delta\gets \delta_1.$ That is, $\delta$ is set to the maximum value s.t.~$a_{\xi_j} + \delta_1 w_{\xi_j}^\top y^*_{\zeta_j} - a_{\xi_i} - \delta_1 w_{\xi_i}^\top y^*_{\zeta_i}\le d(u_{\xi_i}, u_{\xi_j})\ \forall i, j$ s.t. $a_{\xi_j} - a_{\xi_i}< d(u_{\xi_i}, u_{\xi_j})$. Thus, a fairness constraint activates. Let $a_{\xi_i}(\tau)$ be the value of $a_{\xi_i}$ at time $\tau$ in Algorithm~\ref{alg:water}. Define the state of the algorithm as the values $a_{\xi_i}(\tau)$ of each agent. Then, the algorithm defines a path through $\mathbb R^{\Xi}$ over time, which is monotonically increasing with $\tau$. Next, each fairness constraint defines a half-space through this space; let $\mathcal H$ be the finite set of these hyperplanes. The half-spaces divide the state space into a finite number of convex polyhedral regions. Case 3 occurs exactly when the current state $a(\tau)$ reaches the boundary of a half-space, and thus the boundary of a region. Since the state is monotonically increasing in $\tau$, and there are a finite number of regions, we see that Case 3 can only occur a finite number of times as well.
\end{proof}

Next, we prove Lemma~\ref{lem:water_fsbl}.

\begin{proof}[Proof of lemma~\ref{lem:water_fsbl}]
    The algorithm begins with the resources $B'_{\xi_i} \gets b_{\xi_i}^\top y^*_{\xi_i}$ that each agent $\xi_i\in \Xi$ consumes in the unfair matching $y^*$. Since $y^*$ is feasible, as long as the matching $x$ consumes less resources than $y^*$, then $x$ is feasible as well. In particular, we show that at all stages, $\sum_{\xi_i\in \Xi}b_{\xi_i}^\top y^*_{\xi_i} \ge \sum_{\xi_i\in \Xi}b_{\xi_i}^\top x_{\xi_i}$.
    
    Whenever $x$ is updated in line~\ref{alg:water:allocated}, $B'_{\xi_i}$ is updated immediately afterward in line~\ref{alg:water:update_resources}. The update reduces $B'_{\xi_i}$ by exactly the amount of resources consumed by all agents $j\ne i$ that are currently filling along the vector $y^*_{\xi}$ (i.e., agents with $\zeta_j = \xi_i$), plus the amount of resources that $\xi_i$ itself would consume or renounce. Furthermore, the step size $\delta$ determined in subroutine \textsc{StepSize} is bounded by $\delta_2 = $ max value s.t. $B'_{\xi_i} - \delta(b_{\xi_i}+\sum_{j\ne i:\zeta_j = \xi_i}b_{\xi_j})^\top y^*_{\xi_i} \ge 0$. Thus, the remaining capacity $B'$ is always nonnegative. Summing over all agents, the total resources consumed is bounded above by the change in $B'$ throughout the algorithm. We conclude that $\sum_{\xi_i\in \Xi}b_{\xi_i}^\top y^*_{\xi_i}\ge \sum_{\xi_i\in \Xi}B'_{\xi_i} \ge \sum_{\xi_i\in \Xi}b_{\xi_i}^\top x_{\xi_i}$.
\end{proof}

Next, we prove Lemma~\ref{lem:water_matching}.

\begin{proof}[Proof of Lemma~\ref{lem:water_matching}]
    We wish to show that each $\xi_i\in \Xi$ has $||x_{\xi_i}||_{1} \le 1$. For each iteration $k = 1,\dots, K$ of the while loop in line~\ref{alg:water:while}, let $\delta_k$ be the value of $\delta$ computed in \textsc{StepSize}, and let $y^*_k$ be the value of $y^*_{\zeta_i}$ in that iteration. Then, $x_{\xi_i} = \sum_{k = 1}^K \delta_k y^*_k$. Since $||y^*_k||\le 1$ for each $k$, we have $||x_{\xi_i}||_k = \sum_{k=1}^K \delta_k||y^*_k||_1 \le \sum_{k=1}^K\delta_k = 1$.
\end{proof}

Next, we prove Lemma~\ref{lem:water_fair}.

\begin{proof}[Proof of Lemma~\ref{lem:water_fair}]
    Consider any two agents $\xi_\alpha, \xi_\beta$.  Let $\tau^* = \sup_\tau \{0\le \tau\le 1 \mid a^\tau_{\xi_\alpha} - a^\tau_{\xi_\beta}\le d(u_{\xi_\alpha},u_{\xi_\beta})\}$. If $\tau^*=1$, then we are done. Hence, assume $\tau^*<1$. Then, in all calls of \textsc{SelectFill} at all times $\tau > \tau^*$ for agent $\xi_\beta$, the index $\alpha$ gets added to $J$ in Line \ref{alg:selectfill:J}. By Assumption~\ref{d_assumption_w} on the distance function, the allocation vector $y^*_{\zeta_\alpha}$ satisfies 
    $$d(u_{\xi_\alpha}, u_{\xi_\beta})\mathbf{1}^\top y^*_{\zeta_\alpha}\ge ||w_{\xi_\alpha} - w_{\xi_\beta}||_\infty\mathbf{1}^\top y^*_{\zeta_\alpha}\ge w_{\xi_\alpha}^\top y^*_{\zeta_\alpha} - w_{\xi_\beta}^\top y^*_{\zeta_\alpha}.$$
    Rearranging then gives $w_{\xi_\beta}^\top y^*_{\zeta_\alpha}\ge (w_{\xi_\alpha} - d(u_{\xi_\alpha}, u_{\xi_\beta})\mathbf{1})^\top y^*_{\zeta_\alpha}$. Since $\alpha\in J$, in lines~\ref{alg:selectfill:j*} and \ref{alg:selectfill:zeta} of \textsc{SelectFill} we know that $w_{\xi_\beta}^\top y^*_{\zeta_\beta}\ge w_{\xi_\beta}^\top y^*_{\zeta_\alpha}$, and so we have $w_{\xi_\beta}^\top y^*_{\zeta_\beta}\ge (w_{\xi_\alpha} - d(u_{\xi_\alpha}, u_{\xi_\beta})\mathbf{1})^\top y^*_{\zeta_\alpha}$. Rearranging and multiplying by $\delta$ gives $\delta w_{\xi_\alpha}^\top y^*_{\zeta_\alpha} - \delta w_{\xi_\beta}^\top y^*_{\zeta_\beta} \le \delta  d(u_{\xi_\alpha}, u_{\xi_\beta})\mathbf{1}^\top y^*_{\zeta_\alpha}\le \delta  d(u_{\xi_\alpha}, u_{\xi_\beta})$. Since each agent receives at most 1 total unit of allocation throughout the algorithm, or the sum of the step sizes $\delta$ is less than 1, $a_{\xi_\alpha}^1 - a_{\xi_\beta}^1 \le (a^{\tau^*}_{\xi_\alpha} - a^{\tau^*}_{\xi_\beta}) + d(u_{\xi_\alpha},u_{\xi_\beta}) \le 2d(u_{\xi_\alpha},u_{\xi_\beta})$, giving 1/2-fairness.
\end{proof}

To prove Lemma~\ref{lem:water_optimal}, we first need a helper result showing that whenever agent $\xi_i$ fills from another agent's vector, then fairness must be weakly broken between them.

\begin{lemma}\label{lem:triangle}
    In Algorithm~\ref{alg:selectfill}, suppose $\xi_i$ has $\zeta_i\gets \xi_k$ for some $k\ne i$ in Line~\ref{alg:selectfill:zeta} at some moment of the algorithm. Then, $a_{\xi_k} - a_{\xi_i} \ge d(u_{\xi_i}, u_{\xi_k})$ at that moment of the algorithm.
\end{lemma}
 
\begin{proof}
    Let $\tau$ be the current time of the algorithm. We know that if $\zeta_i\gets \zeta_{\ell_1}$ for some $\ell_1<i$ in Line \ref{alg:selectfill:zeta} of \textsc{SelectFill}, then $\ell_1 = \arg\max_{j\in J}w_{\xi_i}^\top y^*_{\zeta_j}$. Recall that, at time $\tau$, \textsc{SelectFill} is called for each agent $\xi$ in non-decreasing order of $a_\xi$. Thus, \textsc{SelectFill} was called at time $\tau$ on $\xi_{\ell_1}$ before it was called for $\xi_i$. In the earlier execution of \textsc{SelectFill} corresponding to $\ell_1$, if we also had $\zeta_{\ell_1}\gets \zeta_{\ell_2}$ for some $\ell_2<\ell_1$, then we must have had $\ell_2 = \arg\max_{j\in J}w_{\xi_{\ell_1}}^\top y^*_{\zeta_j}$. This reasoning continues until $\zeta_{\ell_{n-1}}\gets \zeta_{\ell_{n}}$ where $\ell_n = k$. We conclude the proof by showing by induction on $m$ that  
    \[a_{\xi_{\ell_{m}}}- a_{\xi_i}\ge d(u_{\xi_i}, u_{\xi_{\ell_n}})\]
    for $m=1,\dots, n$.
    For the base case, we know by Line \ref{alg:selectfill:J} of \textsc{SelectFill} that 
    \[a_{\xi_{\ell_1}} - a_{\xi_i} \ge d(u_{\xi_i}, u_{\xi_{\ell_1}}).\]
    For the inductive step, assume that
    \[a_{\xi_{\ell_m}} - a_{\xi_i}\ge d(u_{\xi_i}, u_{\xi_{\ell_m}}),\label{eq:triangle1}\]
    for some $m=1,\dots, k-1$. By the execution of Line \ref{alg:selectfill:J} of \textsc{SelectFill} corresponding to the $m+1$ iteration of the for loop, we know that
    \[a_{\xi_{\ell_{m+1}}} - a_{\xi_{\ell_m}}\ge d(u_{\xi_{\ell_m}}, u_{\xi_{\ell_{m+1}}})\label{eq:triangle2}.\]
     Adding equations \eqref{eq:triangle1} and \eqref{eq:triangle2} and applying the triangle inequality $d(u_{\xi_i}, u_{\xi_{\ell_m}}) + d(u_{\xi_{\ell_m}}, u_{\xi_{\ell_{m+1}}}) \ge d(u_{\xi_i}, u_{\xi_{\ell_{m+1}}}) $ (Assumption~\ref{d_assumption_t}) gives us
    \[a_{\xi_{\ell_{m+1}}} - a_{\xi_{i}} \ge d(u_{\xi_{i}}, u_{\xi_{\ell_{m+1}}}).\]
\end{proof}
 
We can now prove Lemma~\ref{lem:water_optimal}.
\begin{proof}[Proof of Lemma~\ref{lem:water_optimal}]
    The unfair solution $y^*$ has total value $\sum_{\xi\in \Xi}w_\xi^\top y^*_\xi$. Now, we analyze the value of the final allocation $x$ obtained by Algorithm~\ref{alg:water}, by partitioning the value of $x$ over the agents $\xi\in \Xi$ in a manner different from the natural choice, and comparing to the unfair solution.

    Fix $\xi_i\in \Xi$. Over the course of Algorithm~\ref{alg:water}, resources from $B'_{\xi_i}$ are depleted when the following happens:

    \begin{enumerate}
        \item Some agent $\xi_k$ fills along vector $y^*_{\xi_i}$. For this to happen, $\xi_i$ must not have been added to $\Xi'$ by some prior execution of Line \ref{alg:water:Xi'} of Algorithm~\ref{alg:water}. 
        \begin{enumerate}
            \item If $k=i$, agent $\xi_i$ fills their allocation along vector $y^*_{\xi_i}$ for $\delta$ time in line \ref{alg:water:allocated} of Algorithm~\ref{alg:water}, and a corresponding amount of resources is removed from the resource budget $B'_{\zeta_i}$ in line~\ref{alg:water:update_resources}.  We refer to any allocation of this type as \textit{base} allocation for vector $y^*_{\xi_i}$. 
            \item If $k\neq i$, then similarly agent $\xi_k$ fills their allocation along vector $y^*_{\xi_i}$ for $\delta$ time in Line \ref{alg:water:allocated} of Algorithm~\ref{alg:water}, and  $B'_{\xi_i}$ is decreased by $\delta b_{\xi_k}^\top y^*_{\xi_i}$. We refer to any allocation of this type as \textit{outside} allocation from agent $\xi_k$ to vector $y^*_{\xi_i}$. 
        \end{enumerate}
        \item Agent $\xi_i$ fills their allocation along vector $y^*_{\xi_k}$ for some $k\ne i$. This happens if and only if, in line~\ref{alg:selectfill:zeta} of \textsc{SelectFill}, we set $\zeta_i \gets \xi_k$. In this case, agent $\xi_i$ fills their allocation for $\delta$ time along a different vector $y^*_{\xi_i}$ in Line \ref{alg:water:allocated}. While $x_{\xi_i}$ does not fill along vector $y^*_{\xi_i}$, when computing the resource update in Line~\ref{alg:water:update_resources}, $B'_{\xi_i}$ is still decreased by an amount $\delta b_{\xi_i}^\top y^*_{\xi_i}$, corresponding to the renounced capacity. In particular, note that any resource capacity renounced from $B'_{\xi_i}$ corresponds to an equal amount of resource capacity allocated to agent $\xi_i$ via outside allocation along another vector $y^*_{\xi_k}$. 
    \end{enumerate}

    Fix any $\xi_i, \xi_k\in \Xi$, who may or may not be distinct. Throughout the algorithm, $\xi_k$ will fill along vector $y^*_{\xi_i}$ for a total of at most  $1/\max(b_{\xi_k}^\top y^*_{\xi_i} / b_{\xi_i}^\top y^*_{\xi_i})$ time (where the vector division and max operations are performed elementwise) before a resource is exhausted. In the previous definition, we let $0/0:=1$. We consider in Case 1 below the situation where $\max(b_{\xi_k}^\top y^*_{\xi_i} / b_{\xi_i}^\top y^*_{\xi_i})$ is undefined, i.e., the supremum is achieved for a non-zero numerator and a denominator equal to $0$. In all other cases, if agent $\xi_k$ spends $\delta$ time filling along vector $y^*_{\xi_i}$, we say they have consumed $\delta\max(b_{\xi_k}^\top y^*_{\xi_i} / b_{\xi_i}^\top y^*_{\xi_i})$ \emph{units}
        of allocation along vector $y^*_{\xi_i}$. We categorize the units of allocation along vector $y^*_{\xi_i}$ into the following three types:
    \begin{itemize}
        \item $\alpha$ units of base allocation, i.e., $\xi_i$ filling along $y^*_{\xi_i}$;
        \item $\beta$ units of renounced allocation along $y^*_{\xi_i}$, corresponding to outside allocation between agent $\xi_i$ and a different vector $y^*_{\xi_k}$;
        \item $\lambda$ units of outside allocation of $\xi_{k}\neq \xi_i$ filling along $y^*_{\xi_i}$.
    \end{itemize}    

    Consider the end of the algorithm and recall that we defined $a_\xi := w_\xi^\top x_\xi$ to be the total value accumulated by an agent $\xi\in \Xi$. We next define a vector $a' \in \mathbb{R}^{\Xi}$. For $\xi_i \in \Xi$, start from $a'_{\xi_i}=0$ and increase it as follows:
    \begin{itemize}
        \item If $\xi_i$ fills along vector $y^*_{\xi_i}$ for $\delta$ time, then $\delta w_{\xi_i}^\top y^*_{\xi_i}$ is added to $a'_{\xi_i}$.
        
        \item If $\xi_i$ fills along vector $y^*_{\xi_i}$ for $\delta$ time for some $k \neq i$, then $\frac{1}{2}\delta w_{\xi_i}^\top y^*_{\xi_k}$ is added to $a'_{\xi_i}$ and $\frac{1}{2}\delta w_{\xi_i}^\top y^*_{\xi_k}$ is added to $a'_{\xi_k}$.
    \end{itemize}

    Then, $a'$ can be seen as an alternative partition of the value from $a$. Indeed, the following fact immediately follows from the definition of $a'$.
    
    \begin{myclaim}\label{cl:a-vs-a'}
    $\sum_{\xi_i\in \Xi}a'_{\xi_i}= \sum_{\xi_i\in \Xi}a_{\xi_i}$.
    \end{myclaim}

    Thus, bounding $a$ is equivalent to bounding $a'$. We wish therefore to lower bound the amount of value $a'_{\xi_i}$ assigned to any $\xi_i \in \Xi$ in the partition. We consider two cases:

    \smallskip

    \noindent\textbf{Case 1}: At any time during the algorithm, some agent $\xi_k$ attempts to fill along vector $y^*_{\xi_i}$, but $\max(b_{\xi_k}^\top y^*_{\xi_i} / b_{\xi_i}^\top y^*_{\xi_i})$ is undefined. This case occurs when there is some resource $\ell$ such that $(b_{\xi_i}^\top y^*_{\xi_i})_\ell = B'_{\xi_i\ell} = 0$ at initialization, and some other agent $\xi_k$ with $(b_{\xi_k}^\top y^*_{\xi_i})_\ell>0$ (i.e., $\xi_k$ requires resource $\ell$ when allocated to vector $y^*_{\xi_i}$) wanted to fill along vector $y^*_{\xi_i}$. When this happens, $\xi_i$ is closed and added to $\Xi'$. In this case, we know that $y^*_{\zeta_k} = y^*_{\xi_i}$ at this point in the algorithm. This is only possible if $i$ is added to $J$ in the execution of Line \ref{alg:selectfill:J} in \textsc{SelectFill} corresponding to $\xi_k$, which means that $a_{\xi_i} - a_{\xi_k}\ge d(u_{\xi_i}, u_{\xi_k})$.  Furthermore, we know that $b_{\xi_i}\ne b_{\xi_k}$, and so by Assumption~\ref{d_assumption_b} on the distance function $d$ we have $d(u_{\xi_i}, u_{\xi_k})\ge \underline d$. It follows that $a_{\xi_i}\ge \underline d$, which remains true until the end of the algorithm, as $a_{\xi_i}$ monotonically increases throughout the algorithm. 
    
    Now, the value $a_{\xi_i}$ can be obtained from the base allocation of $\xi_i$, and / or some outside allocation of $\xi_i$ along some other filling vectors $y^*_{\xi_j}$. Note that all value corresponding to the first type of allocation is also added to $a'_{\xi_i}$; of the value corresponding to the second type of allocation, only half is added to $a'_{\xi_i}$. Thus, $a'_{\xi_i}\geq \frac{1}{2}a_{\xi_i}\ge \frac{1}{2}\underline d$. Since for the offline optimum $y^*_{x_i}$ we have $w_{\xi_i}^\top y^*_{\xi_i}\le 1$ (since we assumed $\|w\|_\infty \leq 1$ and $y^*$ is a fractional matching), we conclude that the contribution of $a'_{\xi_i}$ to the total value of $x$ is at least $\frac{1}{2}\underline d$ times $w_{\xi_i}^\top y^*_{\xi_i}$.

    \smallskip

    \noindent\textbf{Case 2}: Case 1 does not apply. Then, we know that $\alpha+\beta+\lambda \ge 1$. That is, the capacity is split between $\alpha$ units of base allocation, $\beta$ units of renounced capacity corresponding to outside allocation between agent $\xi_i$ and other vectors, and $\lambda$ units of outside allocation between vector $y^*_{\xi_i}$ and other agents. We examine the value assigned in the partition $a'$ to $\xi_i$ from each part (note that while the capacity contains $\beta+\lambda$ units of outside allocation, only half of that value is assigned to $a'_{\xi_i}$ in the partition), namely

    \begin{itemize}
        \item The base allocation contributes to $a'_{\xi_i}$ a value $\alpha w_{\xi_i}^\top y^*_{\xi_i}$.

        \item Let $\phi_{\xi_i}$ be the contribution to $a'_{\xi_i}$ of the outside allocations of agent $\xi_i$ along vector $y^*_{\xi_k}$ with $k\neq i$. Observe that $\phi_{\xi_i}\geq \frac{1}{2}\beta w_{\xi_i}^\top y^*_{\xi_i}$. This is by definition of $a'$ and because $\xi_i$ obtains outside value from another vector $y^*_{\xi_k}$ only when $w_{\xi_i}^\top y^*_{\xi_k}\ge w_{\xi_i}^\top y^*_{\xi_i}$. That is, the weight of the outside allocation must be greater than the weight that agent $\xi_i$ could obtain from its own vector $\xi_i$.

        \item Now consider the contribution to $a'_{\xi_i}$ from the $\lambda$ units of outside allocation from agents other than $\xi_i$ who fill along vector $y^*_{\xi_i}$. In order for some other agent $\xi_k$ to be eligible to receive allocation from vector $y^*_{\xi_i}$ in this way at time $\tau$, we know by Lemma~\ref{lem:triangle} that $a^\tau_{\xi_i} - a^\tau_{\xi_k}\ge d(u_{\xi_i}, u_{\xi_k})$. In particular, this holds for the last time $\tau$ that agent $\xi_k$ fills along vector $y^*_{\xi_i}$. Next, note that since $a_{\xi_i}$ increases monotonically throughout the algorithm, $a^\tau_{\xi_i} \le a^1_{\xi_i} = 2\phi_{\xi_i} + \alpha w_{\xi_i}^\top y^*_{\xi_i}$.
        
         Finally, we note that other agents $\xi_k$ may take up resources at a different rate than agent $\xi_i$. The worst case occurs when $(b_{\xi_i}^\top y^*_{\xi_i})_\ell = 1$ and $(b_{\xi_k}^\top y^*_{\xi_i})_\ell = \overline b$ for some resource $\ell$, in which case 1 unit of allocation along vector $y^*_{\xi_i}$ only corresponds to $\frac{1}{\overline b}$ time for agent $\xi_k$ filling along vector $y^*_{\xi_i}$. Then, by using Assumption~\ref{d_assumption_w} on $d(\cdot,\cdot)$, the contribution to $a'_{\xi_i}$ of these $\lambda$ units of outside allocation is at least
        \[\frac{1}{2\overline b}\lambda (w_{\xi_i}^\top y^*_{\xi_i} - \max_{\xi_k\ne \xi_i} d(u_{\xi_i}, u_{\xi_k})) \ge \frac{1}{2\overline b}\lambda (w_{\xi_i}^\top y^*_{\xi_i} - a^\tau_{\xi_i}) \ge \frac{1}{2\overline b}\lambda (w_{\xi_i}^\top y^*_{\xi_i} - 2\phi_{\xi_i} - \alpha w_{\xi_i}^\top y^*_{\xi_i}).\nonumber\]
        
    \end{itemize}
    Putting the three contributions above together, we deduce that the value of $a'_{\xi_i}$ is at least
    \[a'_{\xi_i} \ge \alpha w_{\xi_i}^\top y^*_{\xi_i} + \phi_{\xi_i} + \frac{1}{2\overline b}\lambda (w_{\xi_i}^\top y^*_{\xi_i} - 2\phi_{\xi_i} - \alpha w_{\xi_i}^\top y^*_{\xi_i}).\nonumber\]
    The value $\xi_i$ obtains from the unfair fluid allocation is $w_{\xi_i}^\top y^*_{\xi_i}$, giving an approximation ratio of at least
    \[\frac{\alpha w_{\xi_i}^\top y^*_{\xi_i} + \phi_{\xi_i} + \frac{1}{2\overline b}\lambda (w_{\xi_i}^\top y^*_{\xi_i} - 2\phi_{\xi_i} - \alpha w_{\xi_i}^\top y^*_{\xi_i})}{w_{\xi_i}^\top y^*_{\xi_i}}.\nonumber\]
    This ratio is increasing in $\phi_{\xi_i}$ because $\overline b\ge 1$, and so is minimized when $\phi_{\xi_i}$ is as small as possible. We know that $\phi_{\xi_i}\ge \frac{1}{2}\beta w_{\xi_i}^\top y^*_{\xi_i} \ge \frac{1}{2}(1-\alpha-\lambda)w_{\xi_i}^\top y^*_{\xi_i}$, so the approximation ratio is at least
    \[\frac{\alpha w_{\xi_i}^\top y^*_{\xi_i} + \frac{1}{2}(1-\alpha-\lambda)w_{\xi_i}^\top y^*_{\xi_i} + \frac{1}{2\overline b}\lambda (w_{\xi_i}^\top y^*_{\xi_i} - (1-\alpha-\lambda)w_{\xi_i}^\top y^*_{\xi_i} - \alpha w_{\xi_i}^\top y^*_{\xi_i})}{w_{\xi_i}^\top y^*_{\xi_i}}\nonumber\]
    \[ = \frac{\overline b \alpha - \overline b \lambda + \overline b + \lambda^2}{2\overline b}.\nonumber\]
    Thus, the ratio is monotonically increasing in $\alpha$, and is minimized at $\alpha=0$. Substituting gives us a ratio of
    \[\frac{-\overline b \lambda + \overline b + \lambda^2}{2\overline b}.\nonumber\]
    The global minimum of the ratio is then achieved with $\lambda = \overline b / 2$. Since we require $\lambda\le 1$, we conclude that
    \begin{itemize}
        \item If $\overline b \ge 2$, the minimum is achieved at $\lambda = 1$, leading to an approximation ratio of $\frac{1}{2\overline b}$.
        \item If $\overline b = 1$, then the minimum is achieved at $\lambda = \overline b / 2 = 1/2$, leading to an approximation ratio of $3/8$.
    \end{itemize}
\end{proof}

\section{Proofs Omitted from Section~\ref{sec:algorithm}}\label{sec:app:algorithm}

In this section, we present the proof of Lemma~\ref{lem:dual}.

\begin{proof}[Proof of Lemma~\ref{lem:dual}]
    For any fixed choice of Lagrange multipliers $\mu \geq 0$, we can write
    \begin{align*}
        \flu(F) &=\mathbb E_{F}\begin{bmatrix}
        \max_{X_t\in \mathcal X_t} & \sum_{t=1}^T\sum_{s=1}^Sw_{ts}^\top X_{t}^\top e_{ts} \\
        \text{s.t.} & \sum_{t=1}^T\sum_{s=1}^S b_{ts}^\top X_t^\top  e_{ts}\le TS\rho
        \end{bmatrix}\nonumber \\
        &\le \mathbb E_{F}\left[
        \max_{X_t\in \mathcal X_t} \sum_{t=1}^T\sum_{s=1}^Sw_{ts}^\top X_{t}^\top e_{ts}- \sum_{t=1}^T\sum_{s=1}^S \mu^\top b_{ts}^\top X_t^\top e_{ts} + TS\mu^\top\rho \right].
    \end{align*}
    Factoring out $T$ and computing the expectation over all arrival sequences $\pi$ then gives us 
    \begin{align*}
        \flu(F) &\le T\mathbb E_{\pi\sim F}\left[
        \max_{X_\pi\in \mathcal X_\pi} \sum_{s=1}^S w_{\pi(s)}^\top X_{\pi}^\top e_{\pi(s)} - \sum_{s=1}^S (b_{\pi(s)}\mu)^\top X_{\pi}^\top e_{\pi(s)} + S\mu^\top\rho \right]\\
        &= T\left(\sum_{\pi\in \Pi}p_\pi \max_{X_\pi\in\mathcal X_\pi}\left\{\sum_{s=1}^Sw_{\pi(s)}^\top X_{\pi}^\top e_{\pi(s)} - (b_{\pi(s)}\mu)^\top X_{\pi}^\top e_{\pi(s)}\right\} + S\mu^\top\rho\right). \\
    \end{align*}
    Finally, using the definition of the conjugate gives
    \begin{align*}
        \flu(F) &\leq  T\left(\sum_{\pi\in\Pi} p_\pi w^*_\pi(\mathbf{b}\mu) + S\mu^\top\rho\right) \\
        &= TD(\mu).
    \end{align*}
\end{proof}

\section{Proofs Omitted from Section~\ref{sec:online}}\label{sec:app:online}
 
To prove Lemma~\ref{lem:stopping_time}, we first show that with a suitable step size $\eta$, the updated dual variables $\mu_t$ never exceed the upper bound $\mu^{max}$.

\begin{lemma}\label{lem:stepsize}
    Let $t \leq \tau_A$. Suppose $\mu_{t-1}\le \mu^{max}$ and $\eta\le \frac{\sigma_2}{S\overline b}$. Then $\mu_t \le \mu^{max}$.
\end{lemma}
\begin{proof}
    Fix any $j\in [N]$. We wish to show that $(\mu_t)_j\le \mu^{max}_j$. By Algorithm~\ref{alg:inbatch_lazy}, we have
     \[\mu_t = \arg\min_{\mu\ge 0} \tilde g_t^\top \mu + \frac{1}{\eta} V_h(\mu, \mu_{t-1}). \nonumber\]
    To solve for $\mu_t$, the first-order condition gives us
    \[\eta \left(\sum_{s=1}^S(\rho-b_{ts}^\top X_t^{*\top} e_{ts})\right) + \nabla h(\mu_{t}) - \nabla h(\mu_{t-1}) = 0. \nonumber\]
    Restricting to coordinate $j$ and rearranging gives us
    \[\dot h_j ((\mu_t)_j) = \dot h_j((\mu_{t-1})_j) + \eta \sum_{s=1}^S(b_{ts}^\top X_t^{*\top} e_{ts})_j - \eta S\rho_j.\label{lem:stepsize:derivative}\]
    Define $h^*_j(c) = \max_{\psi \in \mathbb{R}} \{c\psi - h_j(\psi)\}$ for $c \in \mathbb{R}$  as the conjugate function of $h_j$. Then by standard properties of conjugate functions (see, e.g., \citet{boyd2004convex}), we know that $h_j^*$ is $\frac{1}{\sigma_2}$-smooth univariate increasing convex function, such that $\dot h_j^*(\dot h_j(\psi)) = \psi$.
    By the definition of $X^*_t$ in Algorithm~\ref{alg:inbatch_lazy}, and because $0\in \mathcal X_t$, we know that
    \[0 \le \sum_{s=1}^S w_{ts}^\top X_t^{*\top}e_{ts} - \mu_{t-1}^\top b_{ts}^\top X_t^{*\top}e_{ts}. \nonumber\]
    By definition of $\overline w$, this also gives us
    \[0 \le \sum_{s=1}^S \overline w - \mu_{t-1}^\top b_{ts}^\top X_t^{*\top}e_{ts}. \nonumber\]
    Rearranging then gives
    \[S\overline w \ge  \sum_{s=1}^S \mu_{t-1}^\top b_{ts}^\top X_t^{*\top}e_{ts}=  \sum_{s=1}^S \sum_{\hat \jmath=1}^N (\mu_{t-1})_{\hat \jmath} (b_{ts}^\top X_t^{*\top}e_{ts})_{\hat \jmath}. \nonumber\]
    Since $\mu_{t-1}\ge 0$ and $(b_{ts}^T {X^*_t}^T e_{ts})\ge 0$, all summands in the rightmost formula are nonnegative and we deduce
    \[\frac{S\overline w}{(\mu_{t-1})_j} \ge \sum_{s=1}^S (b_{ts}^\top X_t^{*\top}e_{ts})_j.\nonumber\]
    By definition of $\overline b$, we also know that 
    \[\sum_{s=1}^S(b_{ts}^\top X_t^{*\top}e_{ts})_j \le S\overline b. \nonumber\]
    So, using \eqref{lem:stepsize:derivative} we obtain
    \[\dot h_j ((\mu_t)_j) \le \dot h_j((\mu_{t-1})_j) + \eta S \min\left(\frac{\overline w}{(\mu_{t-1})_j}, \overline b\right) - \eta S \rho_j.\label{lem:stepsize:derivative2}\]
    We now consider two cases.

\smallskip 

    \textbf{Case 1}: $\frac{\overline w}{\rho_j}\le (\mu_{t-1})_j \le \mu^{max}_j$. Then, we have $\frac{\overline w}{(\mu_{t-1})_j} - \rho_j\le 0$. By \eqref{lem:stepsize:derivative2}, we have 
    \[\dot h_j ((\mu_t)_j) \le \dot h_j((\mu_{t-1})_j). \nonumber\]
    By convexity of $h_j$, we see that $(\mu_t)_j \le (\mu_{t-1})_j\le \mu^{max}_j$, as desired.

    \textbf{Case 2}: $(\mu_{t-1})_j\le \frac{\overline w}{\rho_j}$. Then, by \eqref{lem:stepsize:derivative2} and convexity of $h_j^*$ we have
    \[(\mu_t)_j = \dot h_j^*(\dot h_j((\mu_t)_j))\le \dot h_j^*\left(\dot h_j\left((\mu_{t-1})_j\right) + \eta S\min\left(\frac{\overline w}{(\mu_{t-1})_j}, \overline b\right) - \eta S\rho_j\right). \nonumber\]
    By strong convexity of $h_j$, we can then write
    \[(\mu_t)_j \le \dot h_j^*\left(\dot h_j\left((\mu_{t-1})_j\right) + \eta S\overline b\right). \nonumber\]
    Using the $\frac{1}{\sigma_2}$-smoothness of $h_j^*$, we have $\dot h_j^*(y) \le \dot h_j^*(x) + \frac{1}{\sigma_2}(y-x)$. Setting $x = \dot h_j((\mu_{t-1})_j)$ and $y = x + \eta S\overline b$, we get
    \[ (\mu_t)_j\le \dot h_j^*(\dot h_j((\mu_{t-1})_j)) + \frac{\eta S\overline b}{\sigma_2} = (\mu_{t-1})_j+ \frac{\eta S\overline b}{\sigma_2}. \nonumber\]
    Finally, by assumption we have $\eta \le \frac{\sigma_2}{S\overline b}$, which combined with the definition of $\mu^{max}_j$ gives us
    \[(\mu_t)_j\le \frac{\overline w}{\rho_j} + 1 = \mu^{max}_j, \nonumber\]
    as desired.
\end{proof}

We can now prove Lemma~\ref{lem:stopping_time} bounding the stopping time $\tau_A$.

\begin{proof}[Proof of Lemma~\ref{lem:stopping_time}]
    We know by Lemma~\ref{lem:stepsize} that $\mu_t\le \mu^{max}$ at every step $t\leq \tau_A$ of the algorithm. In particular, this holds for $t = \tau_A$, proving the first statement.
    
    To prove the second statement, define the sigma algebra $\phi_t = \sigma(\xi_{t+1}, \zeta_t)$, where $\xi_{t+1} = \{w_r, b_r, e_r\}_{r=1}^{t+1}$ and $\zeta_t = \{\tilde x_r\}_{r=1}^t$. Then, for all $t$, we have $X_{t+1}^{*\top}e_{t+1s}\in \phi_t$, while $\tilde x_{t+1}\in \phi_{t+1}$. Let
    \[M_t = \sum_{r=1}^t\sum_{s=1}^S b_{rs}^\top(\tilde x_{rs} - X_r^{*\top}e_{rs}).\nonumber\]
    Then $M_t$ is a martingale with respect to $\phi_t$, as $M_t\in \phi_t$, $\mathbb E_\theta(||M_t||_\infty)$ is bounded, and
    \[\mathbb E_\theta[M_{t+1} - M_t\mid \phi_t] = \sum_{s=1}^S b_{t+1s}^\top \mathbb E\left[\tilde x_{t+1s} - X_{t+1}^{*\top}e_{t+1s}\mid X._{t+1}^{*\top}e_{t+1s}\right] = 0.\nonumber\]
    Now, recall that $\tau_A$ is the first time $\tau$ such that there exists $j$ with
    \[\sum_{t=1}^{\tau}\sum_{s=1}^S(b_{ts}^\top \tilde x_{ts})_j + S\overline b \ge TS\rho_j.\nonumber\]
    Since $b_{ts}^\top \tilde x_{ts}\in \phi_t$, we see that $\tau_A$ is a stopping time with respect to $\phi_t$. By the Optional Stopping Theorem (see, e.g., \citet{durrett2019probability}), we have $\mathbb E_\theta[M_{\tau_A}] = 0$. Thus, we can write
    \[\mathbb E_\theta\left[\sum_{t=1}^{\tau_A}(\tilde g_t)_j\right] = \mathbb E_\theta\left[\tau_A S\rho_j - \sum_{t=1}^{\tau_A}\sum_{s=1}^S(b_{ts}^\top X_t^{*\top}e_{ts})_j\right] = \mathbb E_\theta\left[\tau_A S\rho_j - \sum_{t=1}^{\tau_A}\sum_{s=1}^S(b_{ts}^\top \tilde x_{ts})_j\right].\nonumber\]
    Then, using the definition of $\tau_A$, we have
    \[  \mathbb E_\theta\left[\sum_{t=1}^{\tau_A}(\tilde g_t)_j\right] \le \mathbb E_\theta[\tau_A S\rho_j - TS\rho_j + S\overline b].\nonumber\]
    Rearranging then gives us
    \[\mathbb E_\theta[S(T-\tau_A)]\le \mathbb E_\theta\left[\frac{S\overline b - \sum_{t=1}^{\tau_A}(\tilde g_t)_j}{\rho_j}\right].\label{lem:stopping_time:t-tau}\]
    By definition of the update rule for $\tilde g_t$, we have
    \[\dot h_j ((\mu_{t})_j) \ge \dot h_j((\mu_{t-1})_j) - \eta (\tilde g_t)_j \nonumber\]
    for all $t\le \tau_A$. So, we can write
    \[ - \sum_{t=1}^{\tau_A}(\tilde g_t)_j\le \frac{1}{\eta}\left(\dot h_j((\mu_{\tau_A})_j) - \dot h_j((\mu_0)_j)\right)\le \frac{1}{\eta}\left(\dot h_j(\mu^{max}_j) - \dot h_j((\mu_0)_j)\right). \nonumber\]
    Combining with \eqref{lem:stopping_time:t-tau} gives us
    \[\mathbb E_\theta[S(T - \tau_A)] \le \frac{S\overline b}{\rho_j} + \frac{\dot h_j(\mu^{max}_j) - \dot h_j((\mu_0)_j)}{\eta \rho_j} \nonumber\]
    \[\le \frac{S\overline b}{\underline\rho} + \frac{1}{\eta\underline\rho}||\nabla h(\mu^{max}) - \nabla h(\mu_0)||_\infty, \nonumber\]
    as desired.

\end{proof}

To prove Lemma~\ref{lem:part2}, we need a result by \citet{chen1993convergence}, known as the Three-Point Property.

\begin{lemma}[\cite{chen1993convergence}]\label{lem:3point}
    Let $\eta > 0$, let $f:\mathbb R^N\to (-\infty, \infty]$ be a closed proper convex function, and let $\{\mu_t\}$ be a sequence computed as
    \[\mu_t = \arg\min_{\mu} f(\mu) + \frac{1}{\eta}V_h(\mu, \mu_{t-1}).\nonumber\]
    Then, for any $\mu\ge 0$, 
    \[\eta(f(\mu_{t+1}) - f(\mu))\le  V_h(\mu, \mu_t) - V_h(\mu, \mu_{t+1}) - V_h(\mu_{t+1}, \mu_t) \nonumber\]
    In particular, for $\{\mu_t\}$ computed as in Algorithm~\ref{alg:inbatch_lazy}, we have
    \[\eta\tilde g_t^\top (\mu_{t+1} - \mu) \le V_h(\mu, \mu_t) - V_h(\mu, \mu_{t+1}) - V_h(\mu_{t+1}, \mu_t).\nonumber\]
\end{lemma}

We can now prove Lemma~\ref{lem:part2}.

\begin{proof}[Proof of Lemma~\ref{lem:part2}]
    Define $p\mathcal X = \{y\mid y_{\pi}\in p_\pi \mathcal X_\pi\}\subseteq \mathbb R_+^{\Pi\times U\times V}$. For $y\in p\mathcal X$, define
    \[L(y, \mu): = \sum_{\pi\in\Pi}p_\pi\sum_{s=1}^S\left(w_{\pi(s)}^\top((y_\pi)_{\pi(s)}/p_\pi) - \mu^\top b_{\pi(s)}^\top(y_\pi)_{\pi(s)}\right) + S\mu^\top\rho. \nonumber\]
    Consider the saddle point problem $\min_{\mu\ge 0}\max_{y\in p\mathcal X} L(y, \mu)$. Minimizing over $\mu\ge 0$ gives us the primal:
    \begin{align*}
        \max_{y} & P(y) := \sum_{\pi\in\Pi}p_\pi\sum_{s=1}^S w_{\pi(s)}^\top((y_\pi)_{\pi(s)}/p_\pi) \\
        \text{s.t.} & \sum_{\pi\in\Pi}p_\pi\sum_{s=1}^S b_{\pi(s)}^\top(y_\pi)_{\pi(s)}\le S\rho \\
        & y\in p\mathcal X.
    \end{align*}
    Define
    \[z_t := \arg\max_{z\in p\mathcal X} L(z, \mu_t), \nonumber\]
    \[g_t := \nabla_\mu L(z_t, \mu_t) = -\sum_{\pi\in \Pi} \sum_{s=1}^S b_{\pi(s)}^\top((z_t)_\pi)_{\pi(s)} + S\rho.  \nonumber\]
    Let $\xi_{t} = \{w_r, b_r, e_r\}_{r=1}^{t}$. Also define $\pi_t = \{w_r, b_r, e_r\}$ to be the realized types in the batch at time $t$. Then $\mu_{t+1}\in \sigma(\xi_t)$, i.e., the $\sigma$-algebra over events $\xi_t$. By definition of $\tilde g_t, \overline b, \overline \rho$, we can write
    \[\mathbb E_{\pi_t}||\tilde g_t||_\infty^2 \le 2\left(\mathbb E_{\pi_t}||\sum_{s=1}^S b_{ts}^\top X_t^{*\top}e_{ts}||_\infty^2 + ||S\rho||_\infty^2\right) \le 2S^2(\overline b^2 + \overline \rho^2).\label{lem:part2:eq_gt2}\]
    Note that $\mu_t\in\sigma(\xi_{t-1})$, $g_t\in \sigma(\xi_{t-1})$, and $\tilde g_t\in \sigma(\xi_t)$. Also, $\mathbb E_{\pi_t}[\tilde g_t] = g_t$. Now, for all $0\le \mu\le \mu^{max}$, we can write using Lemma~\ref{lem:3point}
    \begin{align*}
        g_t^\top(\mu_t - \mu) &= \mathbb E_{\pi_t}\left[\tilde g_t\mid\mu_t\right]^\top (\mu_t-\mu) \nonumber\\
        &\le \mathbb E_{\pi_t}\left[\tilde g_t^\top(\mu_t-\mu_{t+1})+\frac{1}{\eta}V_h(\mu, \mu_t) - \frac{1}{\eta}V_h(\mu, \mu_{t+1}) - \frac{1}{\eta}V_h(\mu_{t+1}, \mu_t) \mid \mu_t\right].
    \end{align*}
    By strong convexity of $h$, we can write
    \[g_t^\top(\mu_t - \mu)\le \mathbb E_{\pi_t}\left[\tilde g_t^\top(\mu_t-\mu_{t+1})+\frac{1}{\eta}V_h(\mu, \mu_t) - \frac{1}{\eta}V_h(\mu, \mu_{t+1}) - \frac{\sigma_1}{2\eta}||\mu_{t+1}-\mu_t||_1^2 \mid \mu_t\right].\label{lem:part2:eq1}\]
    Note that using $a^2+b^2\ge 2ab$ for $a,b\in\mathbb R$ and then the Cauchy-Schwarz inequality, we can write
    \begin{align*}
        \frac{\sigma_1}{2\eta}||\mu_{t+1}-\mu_t||_1^2 + \frac{\eta}{\sigma_1}||\tilde g_t||_\infty^2 & \ge ||\mu_{t+1} - \mu_t||_1||\tilde g_t||_\infty \\
        &\ge |\tilde g_t^\top(\mu_t - \mu_{t+1})|.
    \end{align*}
    So, \eqref{lem:part2:eq1} can be continued with
    \[g_t^\top(\mu_t - \mu)
    \le \mathbb E_{\pi_t}\left[\frac{\eta}{\sigma_1}||\tilde g_t||_\infty^2 + \frac{1}{\eta}V_h(\mu, \mu_t) - \frac{1}{\eta}V_h(\mu, \mu_{t+1})\mid \mu_t\right]. \nonumber\]
    Then, using \eqref{lem:part2:eq_gt2}, we have
    \[g_t^\top(\mu_t - \mu)\le \frac{2S^2\eta}{\sigma_1}(\overline b^2 + \overline\rho^2) + \frac{1}{\eta}V_h(\mu,\mu_t) - \frac{1}{\eta}\mathbb E_{\pi_t}\left[V_h(\mu, \mu_{t+1})\mid \mu_t\right]. \nonumber\]
    Taking the expectation with respect to $\xi_{t-1}$, then multiplying by $\eta$, we can write
    \begin{align}
        \mathbb E_{\xi_{t-1}} \left[\eta g_t^\top(\mu_t - \mu)\right]& \le \frac{2S^2\eta^2}{\sigma_1}(\overline b^2 + \overline\rho^2) + \mathbb E_{\xi_{t-1}}[V_h(\mu,\mu_t)] - \mathbb E_{\xi_t}\left[V_h(\mu, \mu_{t+1})\right]. \label{lem:part2:eq2}
    \end{align}
    Now, consider the martingale process 

    \[M_t := \sum_{r=1}^t \eta \tilde g_r^\top (\mu_r - \mu) - \mathbb E_{\xi_{r-1}}[\eta\tilde g_r^{\top}(\mu_r-\mu)]\nonumber\]
    \[ = \sum_{r=1}^t \eta \tilde g_r^\top (\mu_r - \mu) - \eta g_r^{\top}(\mu_r-\mu).\nonumber\]
    We can verify that $M_t\in \sigma(\xi _t)$ and 
    \[\mathbb E[M_{t} - M_{t-1}\mid \xi_{t-1}] = (\mathbb E[\tilde g_t\mid \xi_{t-1}] - g_t)^\top(\mu_t - \mu) = (g_t - g_t)^\top (\mu_t - \mu) = 0.\]
    We can bound its increments using Cauchy-Schwarz, giving us
    $$|M_t - M_{t-1}|\le \eta\left(||\tilde g_t||_\infty + \mathbb E_{\xi_{t-1}}||\tilde g_t||_\infty\right)||\mu_t - \mu||_1.$$
    Since $||\tilde g_t||_\infty\le \overline b + \overline \rho$, we have 
    \[|M_t - M_{t-1}|\le 2(\overline b + \overline\rho)N ||\mu_t - \mu||_\infty.\label{lem:part2:mgle_bound}\]
    Then, by Lemma~\ref{lem:stepsize}, $\mu_t\le \mu^{max}$, given $\mu_{t-1}\le \mu^{max}$. Thus, we can further bound \eqref{lem:part2:mgle_bound} by
    $$|M_t - M_{t-1}|\le 4N(\overline b + \overline\rho)\mu^{max} <\infty.$$
    Thus, the increments are finite, and by the Optional Stopping Theorem, we have $\mathbb E[M_{\tau_A}] = 0$.
    Then we can write \begin{gather}
        \mathbb E\left[\sum_{t=1}^{\tau_A} \eta \tilde g_t^\top(\mu_t - \mu)\right] = \mathbb E\left[\sum_{t = 1}^{\tau_A}\mathbb E_{\xi_{t-1}}[\eta \tilde g_t^\top (\mu_t - \mu)]\right] \nonumber\\
        \le \frac{2S^2\eta^2}{\sigma_1}(\overline b^2 + \overline\rho^2)\mathbb E[\tau_A] + V_h(\mu, \mu_0),\label{lem:part2:eqright}
    \end{gather}
    where the inequality follows from \eqref{lem:part2:eq2}. Next, if we select $\mu = 0$, we can write using the definition of $g_t$ and linearity of $L$ in $\mu$ that
    \begin{align*}
        \sum_{t = 1}^{\tau_A}\eta g_t^\top (\mu_t - \mu) &= \sum_{t = 1}^{\tau_A}\eta \nabla_\mu L(z_t, \mu_t)^\top (\mu_t - \mu) \\
        &= \sum_{t = 1}^{\tau_A}\eta (L(z_t, \mu_t) - L(z_t, \mu)).
    \end{align*}
    Expanding using the definition of $L$ and $D$ gives us
    \begin{align*}
        \sum_{t = 1}^{\tau_A}\eta g_t^\top (\mu_t - \mu)&= \sum_{t = 1}^{\tau_A}\eta \left(L(z_t, \mu_t) - P(z_t) - \mu^\top\left(S\rho - \sum_{\pi\in\Pi}\sum_{s=1}^Sb_{\pi(s)}^\top((z_t)_\pi)_{\pi(s)}\right)\right) \\
        &= \sum_{t = 1}^{\tau_A}\eta \left(D(\mu_t) - P(z_t) - \mu^\top\left(S\rho - \sum_{\pi\in\Pi}\sum_{s=1}^Sb_{\pi(s)}^\top((z_t)_\pi)_{\pi(s)}\right)\right).
    \end{align*}
    Since $\mu = 0$, this is equal to
    $$\sum_{t = 1}^{\tau_A}\eta g_t^\top (\mu_t - \mu)=\sum_{t = 1}^{\tau_A}\eta \left(D(\mu_t) - P(z_t)\right).$$
    Then by convexity of $D$, we have
    \[\sum_{t = 1}^{\tau_A}\eta g_t^\top (\mu_t - \mu)\ge \tau_A\eta\left(D(\tilde\mu_{\tau_A}) - \frac{\sum_{t=1}^{\tau_A}P(z_t)}{\tau_A}\right).\label{lem:part2:eqleft}\]
    \eqref{lem:part2:eqright} and \eqref{lem:part2:eqleft} give us, for $\mu = 0$,
    \[\mathbb E\left[\tau_AD(\tilde\mu_{\tau_A}) - \sum_{t=1}^{\tau_A}P(z_t)\right]\le \frac{2S^2\eta}{\sigma_1}(\overline b^2 + \overline\rho^2)\mathbb E[\tau_A] + \frac{1}{\eta}V_h(\mu, \mu_0).\label{lem:part2:eq_end1}\]
    Next, recall that $\mu_t, z_t\in\sigma(\xi_{t-1})$. If a batch $\pi$ arrives at time $t$, Algorithm~\ref{alg:inbatch_lazy} computes $X^*_t = (z_t)_\pi/p_\pi$ and then randomly samples $\tilde x_{ts}\sim X_{ts}^{*\top}e_{ts}$. Thus, for $t\le \tau_A$, we have
    $$\mathbb E_{\pi_t}\left[\sum_{s=1}^S w_{ts}^\top \tilde x_{ts}\mid \xi_{t-1}\right] = \sum_{\pi\in\Pi}p_\pi\sum_{s=1}^S w_{\pi(s)}^\top(((z_t)_\pi)_{\pi(s)}/p_\pi) = P(z_t).$$
    So, the process
    $$N_t = P(z_t) - \mathbb E_{\pi_t}\left[\sum_{s=1}^S w_{ts}^\top \tilde x_{ts}\mid \xi_{t-1}\right]$$
    is a martingale with finite intervals, and the Optional Stopping Theorem gives us $\mathbb E[N_{\tau_A}] = 0$. Equivalently, we have
    \[\mathbb E\left[\sum_{t=1}^{\tau_A}\sum_{s=1}^S w_{ts}^\top \tilde x_{ts}\mid \xi_{t-1}\right] = \mathbb E\left[\sum_{t=1}^{\tau_A} P(z_t)\right].\label{lem:part2:eq_end2}\]
    Finally, we can combine \eqref{lem:part2:eq_end1} and \eqref{lem:part2:eq_end2} to give
    \[\mathbb E_F\left[\tau_A D(\tilde \mu_{\tau_A}) - \sum_{t=1}^{\tau_A}\sum_{s=1}^Sw_{ts}^\top \tilde x_{ts}\right] \le \frac{2\eta S^2(\overline b^2 + \overline \rho^2)}{\sigma_1}\mathbb E_F[\tau_A] + \frac{1}{\eta}V_h(0, \mu_0), \nonumber\]
    as desired.
\end{proof}

\end{document}